\documentclass[letterpaper,11pt]{article}
\pdfoutput=1

\usepackage{diagbox}
\usepackage{makecell}
\usepackage{booktabs}
\usepackage{tikz,tikz-3dplot}
\usepackage{amsmath}
\usepackage{bbm}
\usepackage{amsfonts}
\usepackage{amssymb}
\usepackage{amsthm}
\usepackage{url,ifthen}
\usepackage[shortlabels]{enumitem}
\usepackage{srcltx}
\usepackage{dsfont}
\usepackage{multirow}
\usepackage{boxedminipage}
\usepackage[margin=1.1in]{geometry}
\usepackage{nicefrac}
\usepackage{xspace}
\usepackage{graphicx}
\usepackage{color}
\usepackage{subfigure}
\usepackage{colortbl}
\usepackage{setspace}
\usepackage{pgfplots}
\pgfplotsset{compat=1.12}
\usepackage{algorithm,algorithmic}
\usepackage[algo2e]{algorithm2e}

\usepackage{xcolor}
\definecolor{DarkGreen}{rgb}{0.1,0.5,0.1}
\definecolor{DarkRed}{rgb}{0.5,0.1,0.1}
\definecolor{DarkBlue}{rgb}{0.1,0.1,0.5}
\definecolor{Gray}{rgb}{0.2,0.2,0.2}

\definecolor{c1}{RGB}{38, 70, 83}
\definecolor{c2}{RGB}{42, 157, 143}
\definecolor{c3}{RGB}{233, 196, 106}
\definecolor{c5}{RGB}{231, 111, 81}
\definecolor{c4}{RGB}{244, 162, 97}

\definecolor{c1}{RGB}{38, 70, 83}
\definecolor{c2}{RGB}{42, 157, 143}
\definecolor{c3}{RGB}{233, 196, 106}
\definecolor{c5}{RGB}{231, 111, 81}
\definecolor{c4}{RGB}{244, 162, 97}
\newcommand\blfootnote[1]{%
  \begingroup
  \renewcommand\thefootnote{}\footnote{#1}%
  \addtocounter{footnote}{-1}%
  \endgroup
}

\usepackage{listings}
\lstdefinestyle{mystyle}{
    commentstyle=\color{DarkBlue},
    keywordstyle=\color{DarkRed},
    numberstyle=\tiny\color{Gray},
    stringstyle=\color{DarkGreen},
    basicstyle=\footnotesize,
    breakatwhitespace=false,         
    breaklines=true,                 
    captionpos=b,                    
    keepspaces=true,                 
    numbers=left,                    
    numbersep=5pt,                  
    showspaces=false,                
    showstringspaces=false,
    showtabs=false,                  
    tabsize=2
}
\usepackage[small]{caption}
\usepackage[pdftex]{hyperref}
\hypersetup{
    unicode=false,          
    pdftoolbar=true,        
    pdfmenubar=true,        
    pdffitwindow=false,      
    pdfnewwindow=true,      
    colorlinks=true,       
    linkcolor=DarkBlue,          
    citecolor=DarkBlue,        
    filecolor=DarkBlue,      
    urlcolor=DarkBlue,          
    pdftitle={},
    pdfauthor={},
}

\def\draft{1}

\def\submit{0}

\ifnum\draft=1 
    \def\ShowAuthNotes{1}
\else
    \def\ShowAuthNotes{0}
\fi

\ifnum\submit=1
\newcommand{\forsubmit}[1]{#1}
\newcommand{\forreals}[1]{}
\else
\newcommand{\forreals}[1]{#1}
\newcommand{\forsubmit}[1]{}
\fi

\ifnum\ShowAuthNotes=1
\newcommand{\authnote}[2]{{ \footnotesize \bf{\color{DarkRed}[#1's Note:
{\color{DarkBlue}#2}]}}}
\else
\newcommand{\authnote}[2]{}
\fi

\newtheorem{theorem}{Theorem}[section]

\newtheorem{lemma}[theorem]{Lemma}

\newtheorem{proposition}[theorem]{Proposition}

\newtheorem{definition}[theorem]{Definition}

\newtheorem*{definition*}{Definition}
\newtheorem*{proposition*}{Proposition}

\theoremstyle{definition}
\newtheorem*{example*}{Example}

\newtheoremstyle{example_contd}
{\topsep} {\topsep}%
{}
{}
{\bfseries}
{.}
{1em}
{\thmname{#1} \thmnumber{ #2}\thmnote{#3} (continued)}

\theoremstyle{example_contd}

\newcommand{\chapterref}[1]{\hyperref[ch:#1]{Chapter~\ref{ch:#1}}}

\newcommand{\claimref}[1]{\hyperref[claim:#1]{Claim~\ref{claim:#1}}}

\newcommand{\corollaryref}[1]{\hyperref[cor:#1]{Corollary~\ref{cor:#1}}}
\newcommand{\definitionlabel}[1]{\label{def:#1}}
\newcommand{\definitionref}[1]{\hyperref[def:#1]{Definition~\ref{def:#1}}}
\newcommand{\equationlabel}[1]{\label{eq:#1}}
\newcommand{\equationref}[1]{\hyperref[eq:#1]{Equation~\ref{eq:#1}}}

\newcommand{\factref}[1]{\hyperref[fact:#1]{Fact~\ref{fact:#1}}}
\newcommand{\figurelabel}[1]{\label{fig:#1}}
\newcommand{\figureref}[1]{\hyperref[fig:#1]{Figure~\ref{fig:#1}}}

\newcommand{\tableref}[1]{\hyperref[tab:#1]{Table~\ref{tab:#1}}}

\newcommand{\itemref}[1]{\hyperref[item:#1]{Item~(\ref{item:#1})}}
\newcommand{\lemmalabel}[1]{\label{lem:#1}}
\newcommand{\lemmaref}[1]{\hyperref[lem:#1]{Lemma~\ref{lem:#1}}}
\newcommand{\proplabel}[1]{\label{prop:#1}}
\newcommand{\propref}[1]{\hyperref[prop:#1]{Proposition~\ref{prop:#1}}}

\newcommand{\propositionref}[1]{\hyperref[prop:#1]{Proposition~\ref{prop:#1}}}

\newcommand{\remarkref}[1]{\hyperref[rem:#1]{Remark~\ref{rem:#1}}}
\newcommand{\sectionlabel}[1]{\label{sec:#1}}
\newcommand{\sectionref}[1]{\hyperref[sec:#1]{Section~\ref{sec:#1}}}
\newcommand{\theoremlabel}[1]{\label{thm:#1}}
\newcommand{\theoremref}[1]{\hyperref[thm:#1]{Theorem~\ref{thm:#1}}}

\newcommand{\assumptionref}[1]{\hyperref[ass:#1]{Assumption~\ref{ass:#1}}}

\usepackage[utf8]{inputenc}
\usepackage[T1]{fontenc}
\usepackage{kpfonts}
\usepackage{microtype}

\newcommand{\Esymb}{\mathbb{E}}

\DeclareMathOperator*{\E}{\Esymb}

\renewcommand{\Pr}{\mathrm{Pr}}

\newcommand{\Prob}[1]{\Pr[ #1 ]}

\renewcommand{\hat}{\widehat}

\newcommand{\cD}{{\cal D}}

\newcommand{\cN}{{\cal N}}
\newcommand{\cO}{{\cal O}}

\newcommand{\cX}{{\cal X}}

\renewcommand{\leq}{\leqslant}

\renewcommand{\geq}{\geqslant}

\newcommand{\norm}[1]{\lVert#1\rVert_2}

\newcommand{\R}{\mathbb{R}}

\usepackage{bm}

\newcommand{\ignore}[1]{}

\DeclareMathOperator*{\argmax}{arg\,max}

\renewcommand{\epsilon}{\varepsilon}

\newcommand{\eps}{\epsilon}

\newcommand{\remove}[1]{}

\newcommand{\Ber}{\mathrm{Ber}}

\newcommand{\ind}{\mathbf{1}}

\newcommand{\1}{\ind}
\renewcommand{\cX}{\mathcal{X}}

\newcommand{\rand}{\mathrm{rand}}
\newcommand{\randpol}{\pi_{\rand}}

\newcommand{\obs}{s}
\newcommand{\uobs}{t}

\newcommand{\Vlin}{V_{*}^{\mathrm{lin}}}
\newcommand{\Vprob}{V_{*}^{\mathrm{pr}}}

\newcommand{\parlin}{\mathrm{PAR}^{\mathrm{lin}}(\alpha, \gamma_\obs, \Delta_\alpha, \Delta_{r^2})}
\newcommand{\parprob}{\mathrm{PAR}^{\mathrm{pr}}(\alpha, \gamma_\obs, \Delta_\alpha, \Delta_{r^2})}

\usepackage{fontawesome}
\usepackage{xcolor}
\usepackage{natbib}
\usepackage{url}
\usepackage{xspace}

\title{The Relative Value of Prediction \\ in Algorithmic Decision Making} 
\author{Juan Carlos Perdomo \\ Harvard University}
\date{\today}
\begin{document}

\maketitle
\begin{abstract}
Algorithmic predictions are increasingly used to inform the allocations of goods and interventions in the public sphere.
In these domains, predictions serve as a means to an end. 
They provide stakeholders with insights into likelihood of future events as a means to improve decision making quality, and enhance social welfare. 
However, if maximizing welfare is the ultimate goal, prediction is only a small piece of the puzzle.
There are various other policy levers a social planner might pursue in order to improve bottom-line outcomes, such as expanding access to available goods, or increasing the effect sizes of interventions.

Given this broad range of design decisions, a basic question to ask is: What is the relative value of prediction in algorithmic decision making?
 How do the improvements in welfare arising from better predictions compare to those of other policy levers? 
The goal of our work is to initiate the formal study of these questions. 
Our main results are theoretical in nature. 
We identify simple, sharp conditions determining the relative value of prediction vis-\`{a}-vis expanding access, within several statistical models that are popular amongst quantitative social scientists. Furthermore, we illustrate how these theoretical insights can guide the design of algorithmic decision making systems in practice. 
\end{abstract}
\pagenumbering{gobble}
\pagenumbering{arabic}

\section{Introduction}
Algorithmic predictions are increasingly used in the public sphere to improve the allocation of scarce resources. 
At their core, these prediction algorithms aim to provide decision makers with valuable information regarding the impacts of particular interventions, or the likelihood of future events, in order to determine who will receive a social good.
\blfootnote{Email: jcperdomo@g.harvard.edu}

For instance, in education, over half of US public schools use early warning systems that predict which students are likely to drop out, as a means to target their limited counseling resources to those that need them the most \citep{balfanz2019early, survey}. Similarly, in Israel, officials build risk predictors to determine the likelihood that individuals will develop serious pulmonary disease in order to prioritize people for vaccines \citep{barda2020developing}. 
Beyond these specific examples, this idea that better predictions implies better decisions, and hence higher welfare, is pervasive, and underlies the design of risk prediction tools used in numerous domains (see \sectionref{sec:related_work}).

From a design perspective, however, we see that prediction is really a means to an end, rather than an end in of itself. 
Predictions help inform decisions with the ultimate goal of increasing some context-dependent notion of social welfare (e.g. the number of on-time high school graduates). And, if improving social welfare truly is the goal, there are of course many different policy levers one could experiment with. 
To name a few, apart from optimizing the quality of decision making via better predictions, one could keep the current predictor fixed and focus efforts on expanding access to the available resources (e.g., invest in more counselors so as to intervene on more students). Alternatively, one could also focus on improving the quality of the interventions themselves (e.g., boost the quality of tutoring).

Those tasked with solving these resource allocation problems do not always have the luxury to simultaneously pursue all these various avenues of improvement. Given the broad space of design decisions, we inevitably need to identify the most cost efficient ways of improving overall welfare, and understand: 
\begin{center}
\emph{To what extent is prediction truly the highest order bit in algorithmic decision making?} 
\end{center}
In this work, we initiate the formal study of the relative value of prediction. We aim to develop a principled understanding of how improvements in welfare arising from better prediction compare to those induced through other policy levers such expanding access or enhancing quality.
As a byproduct of our inquiry into the relative value of prediction, we also aim to shed light on a related question: What does it mean for a predictor to be ``good enough'' for decision making? 
At what point do we decide that we no longer need to collect more features or train a more accurate predictor?

Summarizing briefly,  we argue that predictions are sufficient to enable reliable decisions if the relative value of prediction is \emph{small}. That is, if it is relatively more cost efficient to improve welfare by keeping the current predictor fixed and instead focusing efforts on expanding access or pursuing other policy levers. Our theoretical results establish simple and precise criteria for when these conditions hold true.

\subsection{Overview of Framework \& Contributions}

In this paper, we study problems where a social planner is tasked with making a single binary decision for every member of a population, subject to a resource constraint. 
The reader may think about the decision as whether or not an intervention, or social good, is allocated to a particular individual unit. 
Constraints reflect limits on the total amount of goods or interventions that can be allocated. 
While there is a single decision being made per unit, crucially, the effects of these decisions differ across units (alternatively, individuals). The goal of the planner is to efficiently select which units are allocated goods in order to maximize overall welfare.

More formally, we assume that there is a distribution $\cD$ over units represented as pairs $(x_i,w_i)$. Here, $x_i \in \cX$ is the feature representation of unit $i$ and $w_i \in \R$ is the \emph{welfare improvement} achieved by making a positive decision (intervening) on unit $i$.  The planner uses these features to decide on an assignment policy $\pi:\cX \rightarrow \{0,1\}$, where $\pi(x_i)=1$ if unit $i$ is intervened on and is 0 otherwise. 
We formalize the idea that resources are scarce by requiring that $\E[\pi(x)] \leq \alpha$ for some $\alpha \in (0,1)$.
Writing this out as an optimization problem, the planner solves:
\begin{align}
\equationlabel{eq:main_opt_problem}
	\max_{\pi \in \{\cX \rightarrow \{0,1\}\}} \;\E_{(x_i, w_i) \sim \cD}[w_i \pi(x_i)] \;
	\text{s.t} \;  \E_{\cD}[\pi(x)] \leq \alpha.  
\end{align}
Learning plays a central role in these problems since the optimal solution is to first predict the expected welfare improvement $\E[w \mid x]$ and then
allocate positive decisions to the top $\alpha$ fraction of units, $\pi_*(x) = \ind\{ \E[w \mid x] > t(\alpha)\}$, where $t(\alpha)$ is some threshold. In practice, $\E[w \mid x]$ is of course not known and we approximate the solution by finding a predictor $f$ such that $f(x) \approx \E[w \mid x]$.

\paragraph{The Prediction-Access Ratio.} We formalize the relative value of prediction through a technical notion we call the \emph{prediction-access ratio}, or PAR. Put simply, this is the ratio between the marginal improvement in social welfare that comes from improving the predictor $f$, and the improvement achieved by keeping prediction fixed and expanding access: 
\begin{align*}
\emph{\text{PAR}} = \frac{\emph{\text{Marginal Improvement from Expanding Access}}}{\emph{\text{Marginal Improvement from Improving Prediction}}}. 
\end{align*}
The reader can think of the marginal improvement from prediction as the increase in the objective function, $\E[w_i \pi(x_i)]$, that comes from finding a predictor with smaller error,  $\E(w_i - f(x_i))^2$. On the other hand, the improvement from expanding access corresponds to the change in the optimal value of $\E[w_i \pi(x_i)]$ that comes from increasing $\alpha$ in the constraint $\E [\pi(x_i)] \leq \alpha$ and being able to assign goods or interventions to more people.

Calculating this ratio is a fundamental step in deciding which system changes are (locally) the most cost-effective. In particular, to decide whether small improvements in access or prediction are more effective, we simply calculate the following cost benefit ratio:
\begin{align}
\frac{\emph{\text{Marginal Cost of Improving Prediction}}}{\emph{\text{Marginal Cost of Expanding Access}}} \times \emph{\text{PAR}}. 
\equationlabel{eq:cost_benefit_ratio}
\end{align}
The cost-efficiency of expanding access or improve prediction hinges on whether this expression is smaller or larger than 1. 
While costs are relatively simple to calculate across different settings (e.g., it costs $y$ dollars to hire an extra teacher and expand access to tutoring), the prediction-access ratio is more subtle and requires counterfactual analysis (e.g., what would the world look like if my decisions were slightly better versus if I could just treat more people?). 

The main technical contribution of our work is to calculate this ratio in the context of workhorse statistical models that are commonly used amongst quantitative social scientists. 
We start by stating our main result for the case where the effects of interventions are real-valued (e.g., increases in life expectancy from receiving a hip replacement). 
\begin{theorem}[Informal]
Assume that welfare improvements follow a linear regression model where $w$ is real-valued and normally distributed. If the predictor $f$ explains a $\gamma_\obs^2$ fraction of the variance in outcomes $w$, and at most an $\alpha$ fraction of the population can be intervened on, then the prediction-access ratio is equal to $$\frac{\gamma_\obs}{\alpha}.$$
\end{theorem}
Essentially, the theorem states that if resources are very scarce in the sense that $\alpha$ is very small, and the predictor explains a non-trivial fraction of the variance (say $1\%$), then the relative value of prediction is \emph{small}: The gains in social welfare that arise from expanding access far outweigh the gains that arise from improving prediction.\footnote{Note that explaining $1\%$ of the variance means $\gamma_\obs \geq .1$. Simple predictors of life outcomes that use basic features typically explain at least a single digit fraction of the variance \citep{salganik2020measuring}. Resources are typically, but not always, much more heavily constrained. See \sectionref{sec:related_work} for further discussion.} In particular, as long as the costs of expanding access are not $1/\alpha$ times larger than those of collecting more data or more features, then the cost-efficient way of making progress is to focus efforts on treating more people.

Another way to interpret this result is to plug in the value of the prediction-access ratio into the cost-benefit analysis from \equationref{eq:cost_benefit_ratio}. A direct corollary of first theorem is that as long as, 
\begin{align*}
\equationlabel{eq:linear_conditions_english}
	\sqrt{\emph{\text{Current Variance Explained}}} \geq 
\frac{\emph{\text{Marginal Cost of Expanding Access}}}{\emph{\text{Marginal Cost of Expanding Prediction}}} \times \emph{\text{Current Access}}, 
\end{align*}
then the cost-benefit ration from \equationref{eq:cost_benefit_ratio} is greater than 1, implying that the relative value of prediction is small and that expanding access is locally optimal.

As per the previous discussion, this also constitutes a sufficient condition for a predictor to be ``good enough'' for algorithmic decision making: From a welfare perspective, it is efficient to leave the predictor ``as is'' and focus on other aspects of the broader system.
 
To get a better sense of how these insights extend to other settings, we also calculate the prediction-access ratio for a problem where welfare improvements are discrete (e.g., changes in unemployment status), rather than real-valued. 

\begin{theorem}[Informal]
Assume welfare improvements follow a probit regression model where $w = \ind\{z >0 \}$, and $z$ is a real-valued and normally distributed score. If the predictor $f$ explains a $\gamma_\obs^2$ fraction of the variance in $z$, and at most an $\alpha$ fraction of the population can be intervened on, then the prediction-access ratio is equal to, $$\frac{\gamma_\uobs}{\alpha^{1 / \gamma_\uobs^2}},$$
where $\gamma_\uobs^2 = 1  - \gamma_\obs^2$ is the fraction of \underline{unexplained} variance in the underlying score $z$.
\end{theorem}

Since $1 / \gamma_\uobs^2 > 1$ for all $\gamma_\obs \in (0,1)$, the prediction-access ratio is almost always \emph{larger} in this discrete case than in the previous, real-valued setting. Hence, the relative value of prediction is smaller and one should be more willing to expand access vis-\`{a}-vis improving prediction. As before, one can similarly factor costs into the equation, and again invert the bound to identify simple conditions under which it is cost-efficient to expand access vs improve prediction, or vice versa (see \equationref{eq:probit_cutoff} in \sectionref{sec:probit}).

\subsection{Implications \& Limitations} 

At a high level, our main finding is that if a predictor explains even a  small fraction of the variance in target outcomes and resources are significantly constrained, the benefits of expanding access far outweigh the welfare benefits induced by improving prediction.

We believe that these results regarding the relative value of prediction are counter-narrative, if not quite counter-intuitive. It makes sense that if interventions have mostly positive effects, we are better off intervening on two people, rather than painstakingly improving prediction to identify the single person who will benefit the most from an intervention. While prediction is a means to an end, expanding access is an end in and of itself.  

Furthermore, our results show that you do not necessarily need to predict something very accurately, in order to act upon it effectively. Within the probit model, for instance, the prediction-access ratio can be larger than 1 for small values of $\alpha$, \emph{regardless} of the predictive value ($r^2$ value) of the model (see \sectionref{sec:probit} and \figureref{fig:par_probit}). Structural improvements like expanding access can have significantly larger ``bang for buck'' relative to the gains achieved from building a better prediction system, even if predictions only explain a vanishingly small fraction of outcomes. 

Crucially, however, our results and implications are limited to domains where: $(a)$ the interventions allocated can be truly though of as social \emph{goods} (and not harms), $(b)$ the social planner has an accurate sense of costs, and $(c)$  prediction problems are well-specified.

\paragraph{Allocating Goods vs Harms.} The models we study in this paper satisfy the property that there is a large fraction of the population for which welfare improvements are positive, $\Prob{w_i > 0} = b \gg 0$. Furthermore, all our bounds for the prediction-access ratio only hold in the regime where $\alpha < b$ ($\alpha$ = level of access). And, if $\alpha < b$, one is (mostly) expanding access to interventions for people who would benefit from them.
We believe that several important resource allocation problems like vaccination campaigns in public health, cash transfers in development economics, and targeted tutoring campaigns in K-12 education fall into this paradigm where treatment effects are mostly positive. 
However, our results are by no means universal and they do not apply to problems where the effects of interventions are less clear.
For instance, if predictions are used to prioritize patients for invasive medical exams, these may prevent future disease, but also cause harm through a misdiagnosis. 
The relevant theoretical analysis in these settings likely differs significantly from ours, and we hope it will be the subject of future work.\footnote{One major difference here is that predictions are valuable not just because they identify people from whom interventions are welfare-enhancing, but also rule out people for whom the intervention is harmful.}

\paragraph{Understanding Costs.} Our analysis only applies to domains where the planner has an accurate sense of the costs of the relevant policy levers. Note that while expanding access is expensive, so is improving prediction. 
In particular, the salient way to improve a predictor in social settings is to collect more features for the people in your population (i.e, perform additional in-person surveys, medical tests) rather than changing the algorithm, or spending money on additional compute. These measurement endeavors come with significant privacy and financial costs.

Here, we treat the costs as known, domain-specific functions (see discussion in \sectionref{subsec:par_linear}) and focus our efforts on analyzing the relative \emph{welfare} impacts of prediction and access. 
This decision is largely motivated by the fact that policy analysts in charge of designing targeting programs often have extensive experience understanding costs, but may lack the tools necessary to evaluate the counterfactual impacts of access versus prediction. 
Nevertheless, we believe it's an important direction for future work to develop new machinery to better estimate costs in domains where these may be less clear. 

\paragraph{Models (Mis)Specification.} Lastly, our insights regarding the limits to prediction are shown in a simple, well-specified setting --favorable to prediction-- where a planner is able to directly estimate the effects of interventions. In the language of potential outcomes \citep{neyman1923applications, rubin1974estimating}, the welfare improvements $w_i= y_i(1) - y_i(0)$ are treatment effects: the difference between outcomes under treatment, $y_i(1)$, and no treatment, $y_i(0)$. 
Directly predicting heterogenous treatment effects is in general a difficult problem that requires strong assumptions both on the data \emph{generating} process $\cD$ as well as the prediction model class $\mathcal{F}$ \citep{pearl2009causality,hardtrecht2022patterns}. 
In practice, one should be cautious when applying these insights to domains that may not strictly follow the causal models we study.

Given the lack of previous work in this area, we view these simple and foundational statistical models as a natural starting point to establish a new, formal language around prediction in social systems. We hope it may enable future research that develops a more general theory that analyzes quantities like the prediction-access ratio under broader assumptions on the underlying population, as well as under different choices of social objective functions.\footnote{In particular, as per \equationref{eq:main_opt_problem}, we adopt average outcomes as the main objective. However, this choice is by no means the only one. It is natural to also study other objectives such as maximizing the minimum outcome.}


\section{Motivating Applications and Related Work}
\sectionlabel{sec:related_work}

Our work lies at the intersection of various threads of research from computer science, economics, and related fields. We briefly discuss some of these connections and list examples of resource allocation problems from the applied literature that motivate our framework. 

\paragraph{Economics.} This idea of using predictions to guide the allocation of public resources has been studied extensively within the \emph{prediction policy problems} literature. 
Introduced by \cite{kleinberg2015prediction}, the term refers to policy questions, which do not directly require causal inference, and can be solved via pure prediction techniques. 
This framing has led to several impactful analyses of problems in health policy \citep{obermeyer2019dissecting}, legal decision making \citep{kleinberg2018human}, and hiring \citep{chalfin2016productivity}. 
A parallel line of work by Susan Athey and collaborators \citep{athey2017beyond,athey2018impact,athey2019machine,athey2021policy} aims to establish further connections between machine learning, causality, and applied policy problems. 
On the methods side, \cite{bhattacharya2012inferring}, as well as \cite{manski2004statistical,manski2013diagnostic}, and \cite{kitagawa2018should}, study related statististical questions regarding the design of optimal targeting and testing mechanisms under budget constraints. And,~\cite{viviano2023policy} analyzes the design of targeting policies under network effects.
Our work contributes to this  agenda by introducing a formal framework studying how improvements in prediction compare to other structural avenues for improving welfare. 

\paragraph{Computer Science.} The statistical targeting problem we consider is related to contextual bandits with knapsack (resource) constraints: a decision maker sees features describing a unit, makes a binary decision, and observes a reward (welfare improvement) for the limited units that were intervened on. \cite{agrawal2016linear}, \cite{slivkins2023contextual}, and \cite{verma2021censored} design sublinear regret algorithms for versions of this problem. 
These papers treat the design parameters (resource constraints) as fixed, whereas we view them as a design decision. 
Our work is also related to the literature on online calibration in that it explicitly aims to design predictions that have strong guarantees with respect to the goals of downstream decision makers \citep{foster1998asymptotic,foster2021forecast,noarov2023high}.
Recently,~\cite{liu2023actionability} analyzed the extent to which accurate predictions of future outcomes can help a decision maker looking to choose from multiple different interventions. Relative to their work, we consider a different setting where there is just one treatment, the prediction task is well-specified, and there are a limits on who may be intervened on.  

\subsection{Applications \& Examples of Targeting Systems}

\paragraph{Education.} Apart from the early warning systems example mentioned previously, there are a number of other resource allocation problems in education that follow our framework including: predicting teacher value added in order to use limited school funds most efficiently \citep{chalfin2016productivity}, and predicting youth gun violence to allocate limited spots in mentoring programs \citep{rockoff2011can}. 

\paragraph{Healthcare.} The use of individual risk predictors to determine treatment effects and prioritize patients for interventions is ubiquitous throughout healthcare.
In Israel, one of the major public health agencies regularly predicts outcomes such as COVID-19 mortality risk \citep{barda2020developing}, or the likelihood of hepatitis C \citep{leventer2022using}, to prioritize patients for medical attention. 
\cite{bhattacharya2012inferring} study predicting the prevalence of malaria in order to distribute bug nets amongst households in Kenya. 
In some of these settings, like bug nets or vaccination campaigns, resources are significantly constrained relative to the total size of the population, especially during the initial rollout of the program \citep{unicef}. 

\paragraph{Government.}
The 1993 Unemployment Compensation Amendments Act \citep{clinton} \emph{requires} that U.S. states use profiling tools to predict the likelihood someone will remain unemployed after exhausting regular benefits in order to decide who will receive additional job training resources. Similar ideas are applied within development initiatives. For instance, \cite{aiken2022machine} study the possibility of targeting cash transfers to the poorest members of society by using satellite imagery to predict poverty. 
\cite{black2022algorithmic} analyze issues of fairness and efficiency when using algorithmic predictions to target tax audits. 
Researchers have also studied the use of machine learning to efficiently target building safety inspections \citep{johnson2023improving, hino2018machine, glaeser2016crowdsourcing}. 

Throughout these domains, there is often a single good or intervention being allocated, predictions are used to figure out which units would experience the largest welfare improvements if acted upon, and the decision maker allocates interventions to those with the highest predicted effects as per the targeting problem formalized in \equationref{eq:main_opt_problem}.

\section{Linear Regression Model \& Real Valued Outcomes}
\sectionlabel{sec:linear}
In this section, we present our first set of main results analyzing the prediction-access ratio for the case where welfare improvements follow a linear regression model. 

 Linear regression is the workhorse method of analysis amongst economists, sociologists and other quantitative social scientists. It has been applied, with varying degrees of success \citep{freedman1999association}, to model causal effects (welfare improvements) in consequential domains such as education \citep{angrist1999using,yule1899investigation} and public health \citep{chandler2011predicting}. Furthermore, amongst theoreticians, linear models are the \emph{drosophilia} of formal inquiry into more complex mathematical questions. According to  Pearl, ``many concepts and phenomena in causal analysis were first detected, quantified, and exemplified in linear
structural equation models before they were understood in full generality and applied to nonparametric problems'' \citep{pearl2013linear}. 
As such, they are a natural starting point to study the prediction-access ratio.

\subsection{Model Definition and Technical Preliminaries}

\begin{definition}[Linear Regression Model]
\definitionlabel{def:model_linear}
We say that welfare improvements follow a linear regression model if
\footnote{Note that \definitionref{def:model_linear} is property of the data \emph{generating} process, not just the model class used for prediction. As per the earlier discussion, welfare improvements $w_i$ are treatment effects (causal quantities) and prediction of causal quantities requires assumptions on the underlying data generating process \citep{pearl2009causality}.}

\begin{align}
	w_i  = \langle x_i,\; \beta \rangle + \mu \quad \text{ where } x_i \sim \cN(0, I),
\end{align}
and $\beta \in \R^d$ is a vector of unknown coefficients. 
Here, $\mu = \E[w_i]$ is the average welfare improvement and $\norm{\beta}^2 = \E[(w_i - \mu)^2]$ captures the heterogeneity in outcomes.\footnote{The condition that $x_i \sim \cN(0, I)$ is purely for notational convenience and comes with no further loss of generality.
In particular, our results also apply to the case where $x \sim \cN(0, \Sigma)$, (the features are now correlated) by replacing $\beta$ with $\Sigma^{1/2} \beta$. Since if $x \sim \cN(0, \Sigma)$ then $\langle x,\beta \rangle \stackrel{d}{=} \langle x', \beta' \rangle$ where $x' \sim \cN(0, I)$ and $\beta' = \Sigma^{1/2}\beta.$}\end{definition}

 While outcomes $w_i$ are deterministic functions of $x_i$, the social planner does not observe the full vector of features, but rather only a subset. In particular, we assume that features $x$ are partitioned into $x = (x_\obs, x_\uobs) \in \cX_\obs \times \cX_{\uobs}$: a set of observed features $x_\obs$ and unobserved features $x_\uobs$. 
For example, the impacts of educational interventions on the likelihood of on-time graduation are the function of many variables. Some of these variables (e.g., student demographics, test scores, attendance rate) are observed, and some are typically unobserved (e.g., level of at-home instruction and parental support).
 
 Using the available data, the social planner aims to maximize average welfare, subject to a resource constraint that at most an $\alpha$ fraction of the population can be treated. We formally define the planner's targeting problem below:

 \begin{definition}[Planner's Targeting Problem]
 \definitionlabel{def:targeting_problem}
  Given access to the observable features $x_s$, the planner solves the following optimization problem: 
 \begin{align}
	\equationlabel{eq:policy_lr}
	\max_{\pi \in \{\cX_\obs \rightarrow \{0,1\}\}} \quad &\E_{(x_{i,\obs}, w_i) \sim \cD}[w_i \cdot \pi(x_{i, \obs})] \\
	\text{subject to}\quad & \E_\cD[\pi(x_{i, \obs})] \leq \alpha  \notag.
\end{align}
 \end{definition}

In this regression setting, we assume that given the partition of features into observables and unobservables,
\begin{align*}
w_i  = \langle x_i,\; \beta \rangle + \mu =  \underbrace{\langle x_{i, \obs},\; \beta_\obs \rangle + \mu}_{\emph{\text{Observable}}} + \underbrace{\langle x_{i, \uobs},\; \beta_\uobs \rangle}_{\emph{\text{Unobservable}}} ,
\end{align*}
the planner has enough data that they can learn the coefficients $\beta_\obs$, corresponding to the observable components $x_\obs$ exactly. In practice, one would estimate $\beta_\obs$ in finite samples through procedures like ordinary least squares and find an estimate $\hat{\beta_{\obs}}$ such that $\norm{\beta_s - \hat{\beta_\obs}} \leq \cO(\sqrt{\mathrm{dim}(\beta_s)/n})$. 

However, for the class of problems we are interested in, data typically comes from administrative databases collected by some centralized government authority,  
where the number of samples $n$ is typically much larger than the number of features, $\dim(x_\obs)$. For example, in the Wisconsin early warning system studied in \cite{perdomo2023difficult}, there are about 40 features and $n$ is over 300k (all public school students in Wisconsin). 
Given these relevant scales, the salient axis for improving prediction in these settings is to collect more features (e.g conduct more tests, increase $\dim(x_\obs)$) rather than collecting more samples, since finite sample errors are essentially negligible. We reflect this reality in our model to streamline our presentation.


\paragraph{Technical Preliminaries.} We adopt the following notation. Let,
\begin{align*}
	\phi(z) = \frac{1}{\sqrt{2\pi}} \exp\left(-\frac{z^2}{2}\right) \text{ and } \Phi(t) = \int_{-\infty}^t \phi(z) dz,
\end{align*}
be the PDF and CDF for a standard Normal random variable, $\cN(0,1)$. Lastly, define $\Phi^{-1}(\alpha) = \inf\{t \in \R: \alpha \leq \Phi(t)\}$ to be the associated quantile function. 
The optimal assignment policy $\pi^*_{\obs}$ for the planner's targeting problem can be easily expressed in terms of these quantities. 
\begin{lemma}
\lemmalabel{lemma:opt_policy_linear}
Assume $w_i$ satisfy the linear regression model (\definitionref{def:model_linear}) and that $x = (x_\obs, x_\uobs)$. The optimal assignment policy $\pi^*_\obs$ solving the planner's targeting problem (\definitionref{def:targeting_problem}) is equal to:
\begin{align}
\label{eq:regression_opt_policy}
	\pi^*_\obs(x_{i, \obs}) = \1\{ \langle x_{i,\obs},\; \beta_s \rangle \geq \Phi^{-1}(1-\alpha)\norm{\beta_\obs}\} 	\cdot \1 \{ \langle x_{i,\obs},\; \beta_s \rangle + \mu \geq 0\}.
\end{align}
\end{lemma}

\begin{proof}
In general, the optimal solution to the planner's targeting problem is to intervene on the top $\alpha$ fraction of the population, conditional on those expected welfare effects being positive (see \propref{prop:general_opt_policy}):
\begin{align}
\label{eq:general_opt_policy}
	\pi^*_\obs(x_{i,\obs})  = \1\{E[w_i \mid x_{i,\obs}] > F^{-1}_{\obs}(1-\alpha)\} \cdot \1\{E[w_i \mid x_{i,\obs}] >0\}.
\end{align}
Here, $F_\obs$ is the CDF of the random variable $\E[w_i \mid X_\obs]$, and $F_{\obs}^{-1}$ is its quantile function. In the linear model (\definitionref{def:model_linear}), both of these functions have tractable expressions. Since,
\begin{align*}
w_i = \langle x_{i, \obs},\; \beta_\obs \rangle + \langle x_{i, \uobs},\; \beta_\uobs \rangle + \mu,
\end{align*}
where $\langle x_{i, \obs},\; \beta_\obs \rangle \sim \cN(0, \norm{\beta_\obs}^2)$ and $\langle x_{i, \uobs},\; \beta_\uobs \rangle \sim \cN(0, \norm{\beta_\uobs}^2)$ are independent, then: 
\begin{align*}
	E[w_i \mid x_{i,\obs}] = \langle x_{i,\obs},\; \beta_s \rangle + \mu, \text{ and }
	F_\obs^{-1}(1-\alpha)  =\Phi^{-1}(1-\alpha) \norm{\beta_\obs} + \mu. 
\end{align*}
The result follows from plugging these expressions into \eqref{eq:general_opt_policy}.
\end{proof}

The value $\E[w_i \pi^*_{\obs}(x_{i, \obs})]$ achieved by the policy $\pi^*_\obs$ is intimately tied to the variance explained by the observable features $x_\obs$. We say that the observable features $x_\obs$ have an $r^2$ value of $\gamma_\obs^2$ if the associated predictor, $f(x_\obs) = \langle x_{\obs}, \; \beta_\obs \rangle$, explains a $\gamma_\obs^2$ fraction of the variance in the outcomes $w$.\footnote{Due to the tight correspondence between features $x_\obs$ and predictors $\langle x_\obs, \beta_\obs \rangle$ in this linear model, we will use the phrases, the features (or predictor) have an $r^2$ value of $\gamma_\obs^2$ interchangeably to refer to the condition in \equationref{eq:r2_linear}. For clarity, we also abuse notation and write $r^2(x_\obs)$ instead of the usual $r^2(f)$ for a function $f$.}

\begin{definition}[$r^2$ - Linear Regression]
\definitionlabel{def:def_r2_linear}
For $w_i$ satisfying the linear regression model (\definitionref{def:model_linear}), 
we say that observable features $x_\obs$ have an $r^2$ value of $\gamma^2_\obs$, $r^2(x_\obs)= \gamma^2_\obs \in [0,1]$, if
\begin{align}
\equationlabel{eq:r2_linear}
r^2(x_\obs) = 1 - \frac{\E(w_i - \mu - \langle x_{i,\obs}, \; \beta_\obs \rangle)^2}{\E(w_i - \mu)^2} = \frac{\norm{\beta_\obs}^2}{\norm{\beta}^2} =\gamma_\obs^2.
\end{align} 
\end{definition}

Note that the predictor that uses all the available features $x_\obs = x$, has an $r^2$ value of 1, while a model that uses no features has an $r^2$ value of 0.
With this notation, we can define the value of a policy: the overall welfare improvement achieved as a function of the distribution of the $w_i$, the level of access $\alpha$, and the variance explained by the observed features $x_i$.

\begin{proposition}[Value Function, Linear Regression]
\proplabel{prop:value_f_linear}
For $w_i$ satisfying the linear regression model (\definitionref{def:model_linear}), define $\Vlin(\alpha, \gamma_s)$ to be,\footnote{We parametrize the value function in terms of $\gamma_\obs$ instead of the set of observed features $x_\obs$ since the value of the features in this model is exactly captured by $\gamma_\obs$.}
\begin{align}
\equationlabel{eq:max_linear_value}
	\Vlin(\alpha, \gamma_s) =
		\max_{\pi \in \{\cX_\obs \rightarrow \{0,1\}\}} \E_{(x_{i,\obs}, w_i) \sim \cD}[w_i \cdot \pi(x_{i,\obs})]  \quad \mathrm{ s.t }\quad & \E_\cD[\pi(x_{i, \obs})] \leq \alpha \text{ and } r^2(x_\obs) = \gamma_\obs^2,
\end{align}
 the value achieved by the optimal policy $\pi^*_\obs$ \eqref{eq:regression_opt_policy} that uses features $x_s$ and has an associated $r^2$ value equal to $\gamma^2_\obs$.
Then, for $\alpha < 1/2$ and $\norm{\beta},\mu > 0$, 
\begin{align}
\equationlabel{eq:linear_value}
	\Vlin(\alpha, \gamma_s) = \alpha \mu + \gamma_\obs \norm{\beta} \phi(\Phi^{-1}(1-\alpha)),
\end{align}
where $\phi(\cdot)$ and $\Phi^{-1}()$ are again the PDF and quantile functions for a standard normal $\cN(0,1)$.
\end{proposition}

In this linear setting, the value function for the optimal policy admits a simple closed-form solution (\equationref{eq:linear_value}) and has an intuitive interpretation. In particular, $\Vlin(\alpha, \gamma_s)$ consists of two terms. The first term, $\alpha \mu$, is exactly the value achieved by a random assignment policy. 
For $w_i \sim \cN(\mu, \norm{\beta}^2)$, if you select who gets treated on the basis of a random coin toss, $\pi_{\rand}(x_{i, \obs}) \sim \Ber(\alpha)$, the expected welfare 
is equal to $\E[w_i \randpol(x_{i, \obs})] = \E [w_i] \E[\randpol(x_{i, \obs})] = \alpha \mu$.
\footnote{Here, $\Ber(\alpha)$ is a Bernoulli random variable with parameter $\alpha$. Under $\pi_{\rand}$, each unit $i$ receives the good, or intervention, independently with probability $\alpha$.}

On the other hand, the second term in \equationref{eq:linear_value} captures the value of predictions that are better than random (and hence have $\gamma_\obs > 0$). Recall that for $w_i = \langle \beta,\; x_i \rangle$ and $x_i \sim \cN(0, I)$, the total variance in the outcomes is $\norm{\beta}^2$ and the variance captured by the features is $\norm{\beta_s}^2$ (since $\langle x_{i,\obs},\; \beta_\obs \rangle \sim \cN(0, \norm{\beta_\obs}^2)$. Therefore, the marginal value of using prediction, versus random assignment, scales with the square root of the variance explained, $\norm{\beta_\obs} = \gamma_\obs \norm{\beta}$, and a quantity that depends on the level of access: $\phi(\Phi^{-1}(1-\alpha))$.
	
\paragraph{When is Prediction Even Necessary.} As a final note before discussing our main result for this section, we pause to point out that if there is very little variance in the outcome variable (i.e., $\norm{\beta}^2$ is small) then, random assignment (i.e., deciding allocations according to $\pi_{\rand}(x_{i, \obs}) \sim \Ber(\alpha)$) is near optimal. Statistical targeting, above and beyond simple random allocations, only makes sense if there is significant variance in the outcomes $w_i$. 

To see this, let $V_\mathrm{random}(\alpha) = \E[w_i \randpol(x_{i, \obs})]$ for $\randpol(x_{i, \obs}) = \Ber(\alpha)$. Then, for $\alpha, \mu, \norm{\beta} > 0$,
\begin{align*}
	\frac{V_\mathrm{random}(\alpha)}{\Vlin(\alpha, 1)} = \frac{1}{1 + \frac{\norm{\beta}}{\mu} \alpha^{-1} \phi(\Phi^{-1}(1-\alpha))}, 
	\end{align*}
where $\Vlin(\alpha,1)$ from \equationref{eq:linear_value} is the value of the policy $\pi^*(x)$ with $r^2(x) =1$ that observes \underline{all} features.
As $\norm{\beta} / \mu \rightarrow 0$,
then $V_\mathrm{random}(\alpha) / \Vlin(\alpha, 1) \rightarrow 1$ and $\randpol$ is optimal. 
While somewhat trite, it is nevertheless an important, common-sense sanity check. Prediction-enabled targeting only makes sense if the variance in the target outcomes $w$ is large relative to the mean.
Since the goal of our work is to study the relative value of prediction -- in settings where prediction makes sense to begin with -- we will assume for the remainder that $\norm{\beta} / \mu$ is not too small.

\subsection{The Prediction-Access Ratio for Linear Regression}
\sectionlabel{subsec:par_linear}

\paragraph{Defining Policy Levers.} In the context of the linear regression model, various policy levers such as improving prediction and expanding access have simple, formal  counterparts. 

\begin{itemize}
	\item \emph{Improving Prediction}: Optimizing prediction in this setting equates to enlarging the set of measured features. This yields a new partition of observed and unobserved covariates $(x_{\obs}, x_{\uobs}) \rightarrow (x_{\obs'}, x_{\uobs'})$ such that if $r^2(x_{\obs}) = \gamma_\obs^2$ then  $r^2(x_{\obs'}) = (\gamma_\obs + \Delta_{r^2})^2$ for some $\Delta_{r^2} >0$. 

	As per our notation from \equationref{eq:max_linear_value}, the relevant improvement in welfare is  equal to:
	\begin{align*}
		\Vlin(\alpha, \gamma_s + \Delta_{r^2}) - \Vlin(\alpha, \gamma_s) \quad \emph{\text{(Marginal Improvement from Improving Prediction})}.
	\end{align*}
	\item \emph{Expanding Access}. This corresponds to keeping the current policy $\pi^*_s$ fixed and slackening the resource constraint  $\alpha$ to $\alpha + \Delta_{\alpha}$. The increase in average welfare is equal to 
		\begin{align*}
		\Vlin(\alpha + \Delta_\alpha, \gamma_s) - \Vlin(\alpha, \gamma_s) \quad \emph{\text{(Marginal Improvement from Expanding Access})}.
	\end{align*}
\end{itemize}
Given current system parameters $(\alpha, \gamma_\obs)$, and proposed improvements $\Delta_{r^2}, \Delta_\alpha \in (0,1)$, the prediction-access ratio, $\parlin$, defined formally below, determines the relative impact of these changes on  welfare: 
\begin{align}
\label{eq:par_linear}
 \parlin=\frac{\Vlin(\alpha + \Delta_\alpha, \gamma_s) - \Vlin(\alpha, \gamma_s)}{\Vlin(\alpha, \gamma_s + \Delta_{r^2}) - \Vlin(\alpha, \gamma_s)}.  
\end{align}
Values of the ratio larger than 1 indicate that a $\Delta_\alpha$ increase in access yields a larger increase in welfare than a $\Delta_{r^2}$ improvement in prediction. The reverse is true if the ratio is smaller than 1.

The following theorem is the main result of this section, effectively identifying this ratio (up to a small constant) for the linear regression model.

\begin{theorem}[Prediction-Access Ratio, Linear Regression]
\theoremlabel{theorem:par_linear}
Assume that $w_i$ satisfy the linear regression model (\definitionref{def:model_linear}) with $\norm{\beta},\mu > 0$. For any $\gamma_\obs, \Delta_{r^2} \in (0,1)$ and $\Delta_\alpha$ satisfying $\alpha + \Delta_\alpha < .05, \Delta_\alpha \in [0, 4 \alpha]$, the following inequalities hold:\footnote{In this theorem statement, we assume upper bounds on $\alpha$ to simplify the resulting analytic expressions. From our analysis in \sectionref{sec:lr_appendix}, we can prove a more general result that holds for larger values of $\alpha$, yet the resulting expressions are significantly less interpretable due to the nature of the quantities involved. Likewise, the constants appearing the bounds can be improved as the expense of assuming tighter bounds on $\alpha$ and $\Delta_\alpha$.} 
\begin{align*}
\frac{1}{4}  \cdot \frac{1}{\alpha}  \left( \frac{\mu}{\norm{\beta} \Phi^{-1}(1-\alpha)}+ \gamma_\obs\right)
  \frac{\Delta_\alpha}{\Delta_{r^2}} 
  \leq \parlin  \leq \frac{1}{\alpha} \left( \frac{\mu}{\norm{\beta} \Phi^{-1}(1-\alpha)}+ \gamma_\obs\right) \frac{\Delta_\alpha}{\Delta_{r^2}}.
\end{align*}
\end{theorem}
\begin{proof}[Proof Sketch]
The proof is a direct calculation using Taylor's theorem and facts about Gaussians. In particular, from Taylor's theorem, we know that 
\begin{align*}
	\Vlin(\alpha + \Delta_\alpha, \gamma_s) - \Vlin(\alpha, \gamma_\obs) & = \frac{\partial}{\partial \alpha} \Vlin(\alpha, \gamma_\obs) \cdot \Delta_\alpha + \cO(\Delta_\alpha^2)\\
	\Vlin(\alpha, \gamma_s + \Delta_{r^2}) - \Vlin(\alpha, \gamma_s) &= \frac{\partial}{\partial \gamma_s} \Vlin(\alpha, \gamma_\obs) \cdot \Delta_{r^2} + \cO(\Delta^2_{r^2}).
\end{align*}
Recall from \propref{prop:value_f_linear} that $\Vlin(\alpha, \gamma_\obs) = \alpha \mu + \gamma_\obs \cdot \norm{\beta} \phi(\Phi^{-1}(1-\alpha))$. Hence, we can directly calculate,
\begin{align*}
	\frac{\partial}{\partial \alpha} \Vlin(\alpha, \gamma_\obs) =  \mu + \gamma_{\obs} \norm{\beta} \Phi^{-1}(1-\alpha) \text{ and }  \frac{\partial}{\partial \gamma_s} \Vlin(\alpha, \gamma_\obs) =  \norm{\beta}\phi(\Phi^{-1}(1-\alpha)).
\end{align*}
The key insight is that $\phi(\Phi^{-1}(1-\alpha))$ is well-approximated by $\alpha \Phi^{-1}(1-\alpha)$ for small values of $\alpha$. Therefore, local changes in $\gamma_\obs$ are a factor of $\alpha$ smaller than local changes in $\alpha$. More specifically, 
\begin{align}
	\Vlin(\alpha + \Delta_\alpha, \gamma_s) - \Vlin(\alpha, \gamma_\obs) &\approx \Delta_\alpha (\mu  + \gamma_\obs \norm{\beta} \Phi^{-1}(1-\alpha)) = \widetilde{\Theta}(1) \notag \\
	\Vlin(\alpha, \gamma_s + \Delta_{r^2}) - \Vlin(\alpha, \gamma_s) & \approx \Delta_{r^2} \norm{\beta} \Phi^{-1}(1-\alpha) \alpha =  \widetilde{\Theta}(\alpha).
	\equationlabel{eq:order_prediction_linear}
\end{align}
 The remainder of the argument deals with making the Taylor approximations precise and controlling the quadratic terms. 
\end{proof}

\paragraph{Discussion.} The main message of this theorem is that, if resources are scarce in the sense that $\alpha$ is small, then the local improvements induced by expanding access are far larger than the improvements in welfare that come from better targeting. 
More specifically, if we fix $\Delta_\alpha= \Delta_{r^2}$ and take $\gamma_\obs$ to at least a constant (say 1\%), the prediction-access ratio is order $1 / \alpha$.\footnote{The prediction-access ratio is order $1/\alpha$ since $\Phi^{-1}(1-\alpha)$ is $\Theta(\sqrt{\log(1/\alpha)})$ for small values of $\alpha$, see \propositionref{prop:Phi-inverse}.} Hence, the amount we should be willing to pay to expand access \emph{increases} the smaller $\alpha$ is. Said otherwise, the relative value of prediction in these targeting problems \emph{decreases} with the level of scarcity.   

Recall from the introduction that to decide which policy lever makes sense,  we only need to factor costs into the equation. Here, we encourage the reader to think of costs in a broad sense. The costs of improving prediction are monetary (e.g., compute bills, labor) as well as \emph{social} (e.g., the privacy costs of people releasing personal information to a central authority, or subjecting individuals to a  means test). A similar comment applies to the costs of other policy levers. However, for the sake of modeling, we can imagine functions $C_\alpha$ and $C_{r^2}$ from $[0,1] \rightarrow \R_{\geq 0}$ that give costs numeric values.\footnote{Formally, these functions return the marginal costs of increasing access from $\alpha$ to $\alpha + \Delta_{\alpha}$ and prediction ($r^2$) from $\gamma_\obs^2$ to $(\gamma_\obs + \Delta_{r^2})^2$. We could thus write them as $C_{\alpha}(\Delta_\alpha; \alpha)$ and $C_{r^2}(\Delta_{r^2}; \gamma_\obs)$ but omit the second argument as it is clear from context.} The decision to expand access vs improve prediction hinges on whether the cost-benefit ratio, formally defined for this linear regression model as, 
\begin{align}
\equationlabel{eq:cost_benefit_linear_regression}
\parlin \times \frac{C_{r^2}(\Delta_{r^2})}{C_\alpha(\Delta_\alpha)},
\end{align}
is larger than 1. \theoremref{theorem:par_linear} essentially states that as long as $C_\alpha(\Delta_\alpha) < \gamma_\obs / \alpha C_{r^2}(\Delta_{r^2})$, expanding access is the cost efficient decision.
To see this in more detail, we can reinterpret the result as establishing a \emph{threshold} for when a set of features is sufficiently ``good'' for targeting. As long as,
\begin{align}
\equationlabel{eq:linear_threshold}
	\gamma_\obs \gtrsim \frac{\Delta_{r^2} C_\alpha(\Delta_\alpha)}{\Delta_\alpha C_{r^2}(\Delta_{r^2})} \cdot \alpha -  \frac{\mu}{\norm{\beta} \Phi^{-1}(1-\alpha)},
\end{align} 
the cost-benefit ratio from \equationref{eq:cost_benefit_linear_regression} is $>$ 1 and expanding access makes economic sense. \footnote{\equationref{eq:linear_threshold} formalizes the condition previously stated informally in the introduction regarding when a predictor is sufficiently good to expand access. In \theoremref{theorem:par_linear}, we assume bounds on $\alpha$ to simplify the presentation of the main result. From our analysis, we can prove a more general result that holds for more general values $\alpha$, yet the resulting expressions are unfortunately not as clean. As per our earlier discussion, we largely neglect the second term in \equationref{eq:linear_threshold} since we typically view $\norm{\beta}$ as much larger than $\mu$, which implies this term is close to 0.}

\paragraph{Visualization.} We examine this threshold visually in \figureref{fig:par_linear}. In particular, in each plot we illustrate the cost benefit ratio (\equationref{eq:cost_benefit_linear_regression}), for various choices of the ratio $C_{r^2}(\Delta_{r^2}) / C_{\alpha}(\Delta_{\alpha})$ where in each case $\Delta_\alpha = \Delta_{r^2} = .01$. 

As predicted by our theory (\theoremref{theorem:par_linear} and \equationref{eq:linear_threshold}), we see that, for small values of $\alpha$, the threshold at which one is indifferent between expanding access vs improving prediction is of the form $\gamma_\obs \propto \alpha$.
 That is, given any pairs of level of access $\alpha$ and prediction $\gamma_\obs$, the cost benefit ratio is above 1 if $\alpha \lesssim \gamma_\obs$ and below 1 if $\alpha \gtrsim \gamma_\obs$. The exact constant depends on the relative costs of prediction and expanding access. As per \equationref{eq:linear_threshold}, changing the relative costs only changes the slope of the line.

A different way to interpret this finding is that in the linear regression model, the optimal values of $(\alpha, \gamma_\obs)$ go ``hand in hand''. That is, if a predictor explains a $\gamma_\obs^2$ fraction of the the variance, then the \emph{efficient} level of access is to set $\alpha \propto \gamma_\obs$: One should increase access until it matches the (square root of) the explained variance. This in effect establishes priorities for how a decision making system should be built out.

\begin{figure*}[t!]
\begin{center}
\includegraphics[width=.47\textwidth]{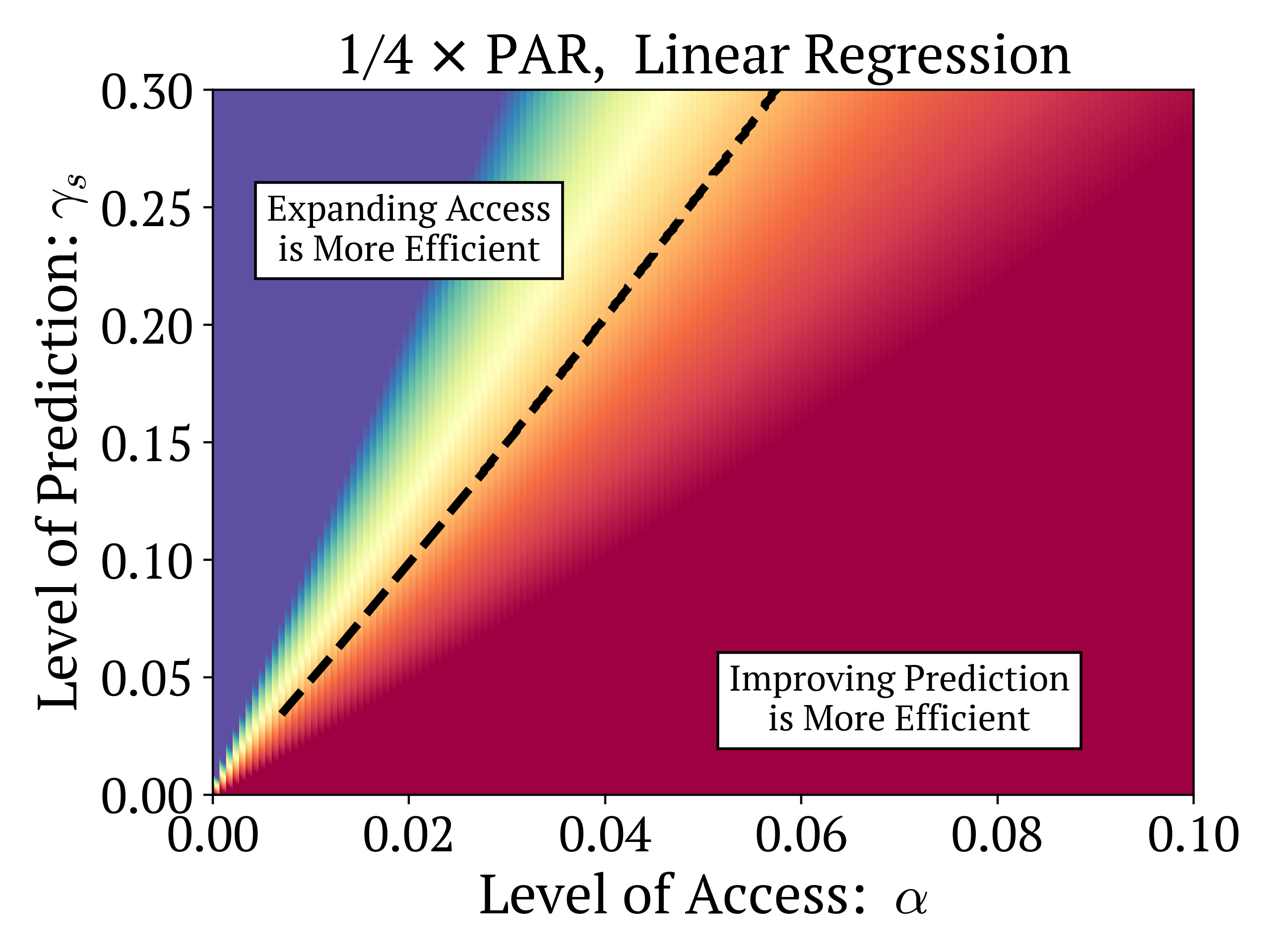}
\includegraphics[width=.47\textwidth]{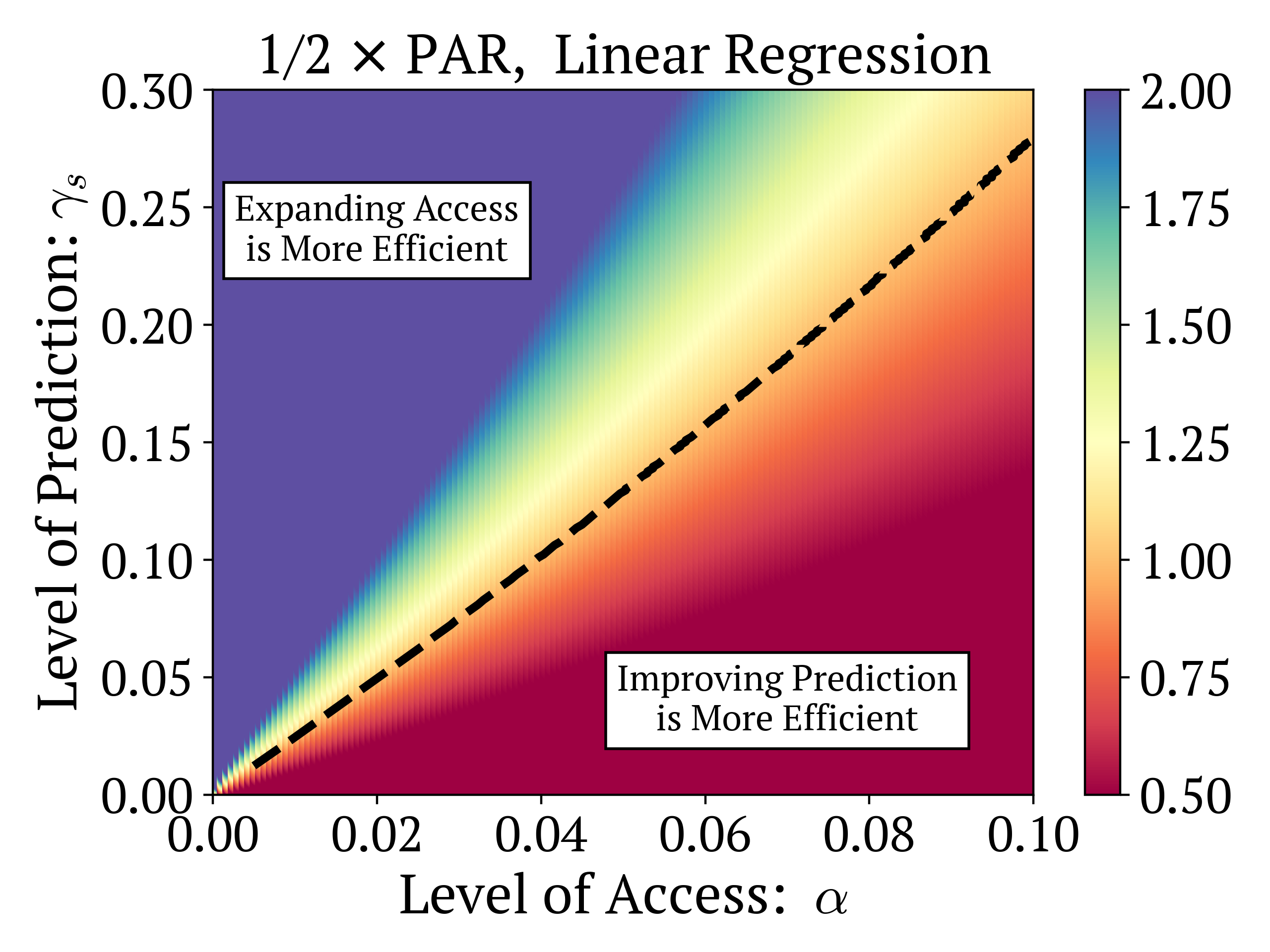}
\end{center}
\vspace{-15pt}
\caption{Visualization of the cost benefit ratio, \equationref{eq:cost_benefit_linear_regression}, for the linear regression model. We compute the ratio for each value of $\alpha$ ($x$-axis) and $\gamma_\obs$ ($y$-axis), exactly via numerical simulation with $\Delta_\alpha = \Delta_{r^2}$. We display its value, clipped to [1/2, 2], via the color bar. We set $ C_{r^2}(\Delta_{r^2}) / C_\alpha(\Delta_\alpha) = 1/4$ on the left and $1/2$ on the right. The black line indicates the set of points for which the ratio is equal to 1. As per \equationref{eq:linear_threshold}, the cutoff is approximate of the form $\gamma_\obs \propto \alpha$, where the slope is determined by the cost ratio. For values $(\alpha, \gamma_\obs)$ above the line, expanding access is relatively cost efficient, whereas improving prediction is efficient for points below the line.}
\figurelabel{fig:par_linear}
\end{figure*}

\paragraph{The Relative Value of Improving Quality.}
So far, we have focused on the relative impacts of prediction and access on welfare. However, there are other policy levers one might experiment with, for instance, improving the quality of the allocated goods. 
One way to formalize this, in this regression model it to consider increasing $\mu = \E[w]$, the average welfare improvement. 

Recall from \propositionref{prop:value_f_linear}	 that the value function of the optimal policy is equal to:
\begin{align*}
	\mu \alpha + \gamma_\obs \norm{\beta} \phi(\Phi^{-1}(1-\alpha)).
\end{align*}
Since the value function is linear in the average quality $\mu$, increasing $\mu$ to $\mu + \Delta_{\mu}$ increases the value function by an additive factor of $\Delta_\mu \alpha$. Hence, in this regression model, increasing the effect size of the average outcome yields an improvement in welfare that is on the same order as improving in prediction. Both are $\widetilde{\Theta}(\alpha)$ (see \equationref{eq:order_prediction_linear}), and significantly smaller than improvements in access, which are $\widetilde{\Theta}(1)$.

\paragraph{Local vs Global Decision Making.} The value of the prediction-access ratio is that it helps guide which kind of \emph{local} improvement is most effective. That is, it answers questions of the form: Given that my current system treats an $\alpha$ fraction of the population and explains a $\gamma_\obs^2$ fraction of the variance in treatment effects, is it better to expand access by 1\% or improve prediction by 1\%? This is motivated by the fact that that there is often significant inertia in the overall design of algorithmic decision making systems.  Large government-run programs, like vaccine rollouts, cannot always be ``restarted from scratch''. Pragmatically speaking, only local changes are viable. 

On a different note, evaluating the impact of local changes is also more mathematically interesting. It is not hard to see that the globally optimal values of $\alpha$ and $\gamma_\obs$ are to set $\gamma_\obs=1$ and $\alpha = \Prob{w_i \geq 0}$ (i.e., predict perfectly and treat anyone who would benefit).

\section{Probit Regression Model \& Discrete Outcomes}
\sectionlabel{sec:probit}
Depending on the way in which we measure outcomes, welfare improvements and treatment effects are often binary (e.g., graduating from high school on-time) rather than real-valued as in the linear regression setting analyzed previously. 
Therefore, in this section we study the prediction-access ratio in the context of a discrete model for outcomes: probit regression.
Like linear regression, probit regression is one of the most popular statistical models used in the quantitative social sciences. 
On a technical level, it belongs to the class of generalized linear models that apply a nonlinearity to the real-valued outputs of a regression, $\langle x, \beta \rangle \in \R$. 
\footnote{The other popular model in this class is logistic regression. Both probit and logistic regression transform the real-valued outputs of a regression model into a score in $[0,1]$ (a probability). The exact transformation is technically different, but morally they tend to provide similar ``bottom-line'' insights when applied in practice \citep{stock2003introduction,scott1997regression}. We focus on the probit model since we found it simpler to analyze.} 

\subsection{Model Definition and Technical Preliminaries}

\begin{definition}[Probit Regression Model]
\definitionlabel{def:model_probit}
We say that welfare improvements follow a probit regression model if 
\begin{align}
	w_i  = \1\{\langle x_i,\; \beta \rangle + \mu > 0\} \quad \text{ for } x_i \sim \cN(0, I),
\end{align}
where again $\beta \in \R^d$ is a vector of coefficients. Furthermore, we denote by  $b = \Pr[w_i > 0]$ the base rate of positive welfare improvements in the population. 
\end{definition}

As before, we assume that features $x$ are partitioned into observed and unobserved components, $(x_\obs, x_\uobs)$ and that the planner collects enough data to learn the subset of coefficients $\beta_\obs$ exactly. In the model outlined above, these parameters can be recovered in finite samples via convex programming at similar $\sqrt{\dim(x_\obs)/n}$ rates \citep{fahrmeir1985consistency}. Given the observable components $x_\obs$, the planner again solves the targeting problem (\definitionref{def:targeting_problem}) where the data now follows the probit model (\definitionref{def:model_probit}).

Due to the nonlinearity (thresholding) being applied, parameters like $\mu$ and $\norm{\beta}$ in this setting no longer directly correspond to natural quantities like mean and variance of $w_i$ as in the linear case. Yet, the optimal policy for the probit case is conveniently the same as before.

\begin{lemma}
\lemmalabel{lemma:opt_policy_probit}
Assume $w_i$ satisfy the probit regression model (\definitionref{def:model_probit}) and that $x = (x_\obs, x_\uobs)$. The optimal assignment policy $\pi^*_\obs$ for the planner's targeting problem (\definitionref{def:targeting_problem}) is equal to:
\begin{align}
\label{eq:probit_opt_policy}
	\pi^*_\obs(x_{i, \obs}) = \1\{ \langle x_{i,\obs}, \beta_s \rangle \geq \Phi^{-1}(1-\alpha)\norm{\beta_\obs}\} .
\end{align}
\end{lemma}
\begin{proof}
Since $\E[w_i \mid x_{i,\obs}]$ is always greater than 0, the optimal policy is simply to treat the units with the highest values of $\E[w_i \mid x_{i,\obs}]$. In this case, 
\begin{align*}
	\E[w_i \mid x_{i,\obs}] = \Pr[w_i = 1 \mid x_{i,\obs}] = \Phi(\mu + \langle x_{i,\obs}, \beta_s \rangle).
\end{align*}
Because $\Phi(z)$ is strictly monotonic in its argument, $z=\mu + \langle x_{i,\obs}, \beta_s \rangle$, the top $1-\alpha$ quantile for $\Phi(z)$ is the same as the top $1-\alpha$ quantile for $z$ itself.
\end{proof}

Given the binary nature of the outcome variable, there is no standard, agreed-upon measure for the coefficient of determination as in the continuous case, with various ``pseudo''-$r^2$ definitions proposed. However, we find it convenient to deal with the same measure as before. 

\begin{definition}[$r^2$ - Probit Regression]
\definitionlabel{def:def_r2_probit}
Assume that $(w_i,x_i)$ satisfy the probit regression model (\definitionref{def:model_probit}). 
We say that the features $x_\obs$ have an $r^2$ value of $\gamma^2_\obs$, $r^2(x_\obs)= \gamma^2_\obs$, if
\begin{align}
r^2(x_\obs) = 1 - \frac{\E(z_i - \mu - \langle x_{i,\obs}, \; \beta_\obs \rangle)^2}{\E(z_i - \mu)^2} = \frac{\norm{\beta_\obs}^2}{\norm{\beta}^2} = \gamma_\obs^2 \in [0,1]
\end{align} 
for $z_i = \langle x_i, \beta \rangle + \mu$ and $w_i = \1\{z_i > 0\}$.\footnote{While our results are meant to be more qualitative than quantitative, one could in principle still calculate this $r^2$ statistic from data by first estimating $\beta_\obs$ and $\mu$ via regression and then computing $\norm{\beta}$ by solving the equation $\Prob{w_i >0} = \Phi^{-1}(\mu\norm{\beta}^{-1})$.}
\end{definition}
 We parametrize the value function for the probit case again in terms of $\gamma_\obs$ and $\alpha$.

\begin{proposition}
\proplabel{prop:value_f_probit}
For $w_i$ satisfying the probit regression model (\definitionref{def:model_linear}), let $\Vprob(\alpha, \gamma_s)$ be
\begin{align}
\equationlabel{eq:max_probit_value}
	\Vprob(\alpha, \gamma_s) =
		\max_{\pi \in \{\cX_\obs \rightarrow \{0,1\}\}} \E_{(x_{i,\obs}, w_i) \sim \cD}[w_i \cdot \pi(x_{i,\obs})]  \quad \mathrm{ s.t }\quad & \E_\cD[\pi(x_{i, \obs})] \leq \alpha \text{ and } r^2(x_\obs) = \gamma_\obs^2 
\end{align}
 the value achieved by the optimal policy $\pi^*_\obs$ (\lemmaref{lemma:opt_policy_probit}) that uses features $x_s$ that have an $r^2$ value equal to $\gamma^2_\obs$ (\definitionref{def:def_r2_probit}). 
Then, 
\begin{align}
\equationlabel{eq:probit_value}
	\Vprob(\alpha, \gamma_s) = \int_{\Phi^{-1}(1-\alpha)}^\infty \Phi\left(\frac{\gamma_\obs z_\obs + \mu \norm{\beta}^{-1}}{\gamma_t}\right) \phi(z_\obs) dz_{\obs}
\end{align}
where $\gamma_t^2 = 1 - \gamma_\obs^2$ is the fraction of unexplained variance in the scores $z_i = \langle x_i, \; \beta \rangle + \mu$.
\end{proposition}

The nonlinear nature of the probit model complicates calculations so that the value functions no longer have simple expressions in terms of analytical functions. Interestingly enough, while $\Vprob$ has no closed-form solution, the prediction-access ratio does. 

\subsection{The Prediction-Access Ratio for Probit Regression}

We define policy levers for the probit model in a similar fashion as before.  

\begin{itemize}
	\item \emph{Improving Prediction}: Expanding the set of measured features $(x_{\obs}, x_{\uobs}) \rightarrow (x_{\obs'}, x_{\uobs'})$ so that if $r^2(x_\obs) = \gamma_\obs^2$ then, $r^2(x_{\obs'}) = (\gamma_\obs + \Delta_{r^2})^2 = \gamma_{\obs'}^2$ for some $\Delta_{r^2} >0$. Formally, this translates to:
	\begin{align*}
		\Vprob(\alpha, \gamma_s + \Delta_{r^2}) - \Vprob(\alpha, \gamma_s) \quad \emph{\text{(Marginal Improvement from Improving Prediction})}.
	\end{align*}
	\item \emph{Expanding Access}. Again corresponds to weakening the resource constraint  $\alpha$ to $\alpha' = \alpha +  \Delta_\alpha$. The change in welfare is equal to: 
		\begin{align*}
		\Vprob(\alpha + \Delta_\alpha, \gamma_s) - \Vprob(\alpha, \gamma_s) \quad \emph{\text{(Marginal Improvement from Expanding Access})}.
	\end{align*}
\end{itemize}
Formally, we define the prediction-access ratio for the probit model to be 
\begin{align*}
	\parprob = \frac{\Vprob(\alpha + \Delta_\alpha, \gamma_s) - \Vprob(\alpha, \gamma_s)}{\Vprob(\alpha, \gamma_s + \Delta_{r^2}) - \Vprob(\alpha, \gamma_s)}.
\end{align*}

\begin{figure*}[t!]
\begin{center}
\includegraphics[width=.47\textwidth]{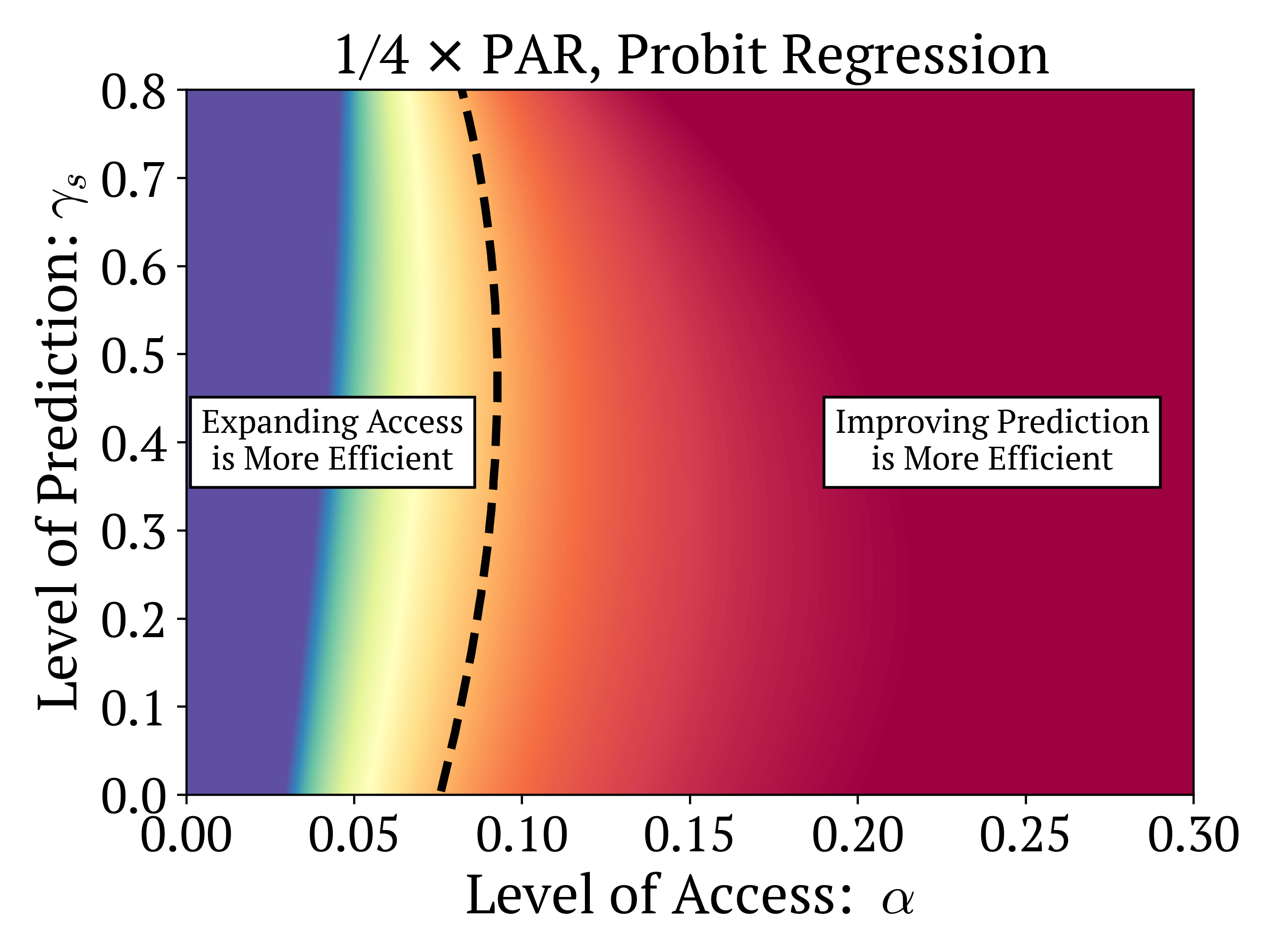}
\includegraphics[width=.47\textwidth]{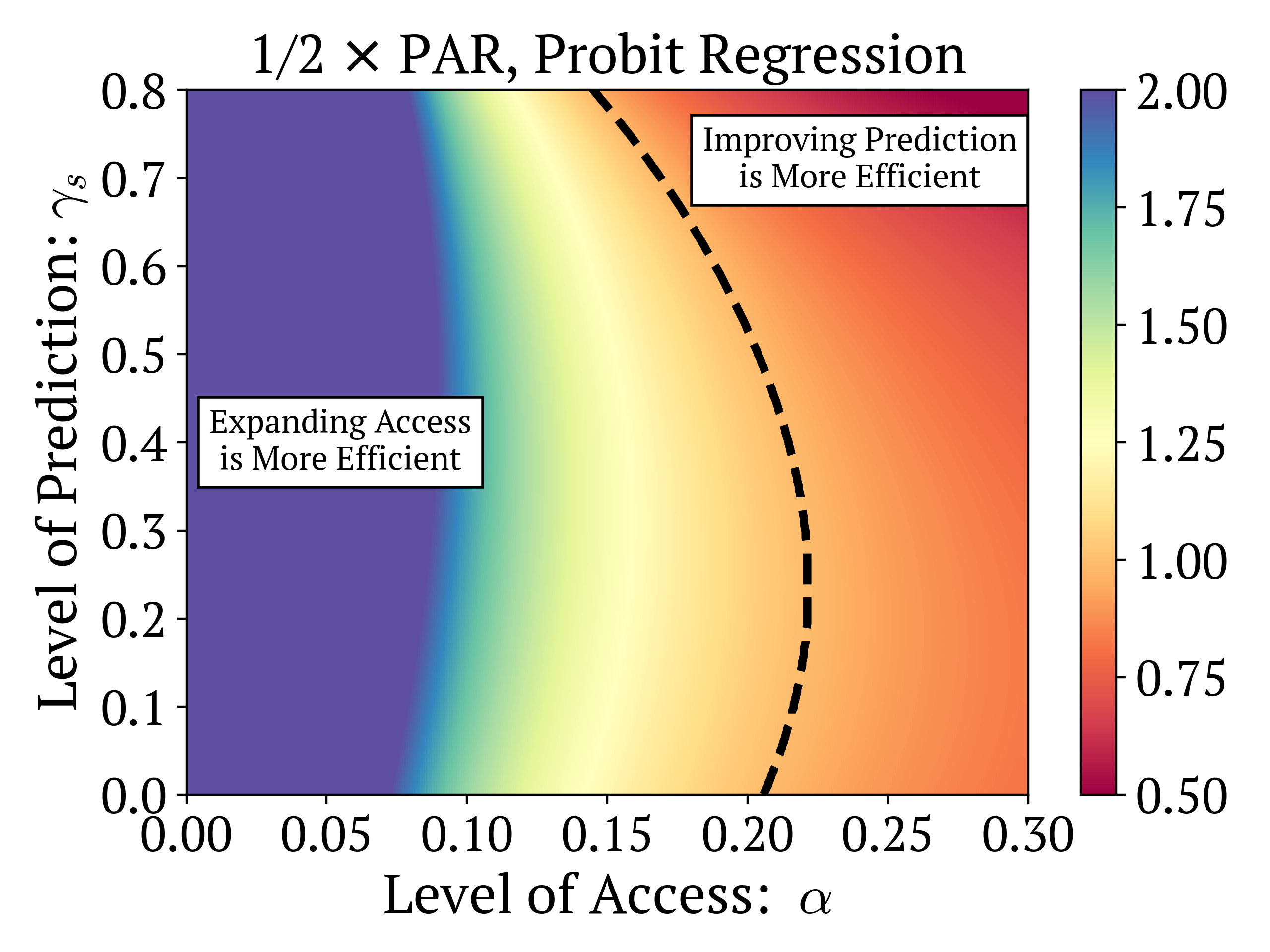}
\end{center}
\caption{Visualization of the cost benefit ratio, \equationref{eq:cost_benefit_probit}, for the probit model. As in \figureref{fig:par_linear}, we compute the ratio numerically with $\Delta_\alpha = \Delta_{r^2}$ and display its value, clipped to [1/2, 2], via the color bar. The black line indicates points for which the ratio is equal to 1. As per \equationref{eq:probit_cutoff}, the threshold  between access and prediction is non-linear, and the ratio is larger than 1 for small $\alpha$, regardless of $\gamma_\obs$.}
\figurelabel{fig:par_probit}
\end{figure*}

\begin{theorem}[Prediction-Access Ratio, Probit Regression]
\theoremlabel{theorem:par_probit}
For any $\gamma_{s} \in (0,1)$, $b = \Pr[w_i = 1] < .1$, and $\eps \in (0,.1)$, let $\epsilon' = \epsilon / (1-\epsilon)$. There exists a threshold $t > 0$, such that for $\max\{\Delta_{r^2}, \alpha\} \leq t$ and $\Delta_{\alpha}\leq \alpha$, the following inequalities hold
\begin{align*}
\parprob &\geq \frac{3}{10} \frac{\Delta_\alpha \gamma_{\uobs}}{\Delta_{r^2}}  \frac{1}{b \Phi^{-1}(1-b)} \left(  \frac{1}{1.01\sqrt{2\pi}}   \frac{1}{\alpha \Phi^{-1}(1-\alpha)}\right)^{1 / \gamma_{\uobs}^2\cdot (1 - \eps)^2}   \\
\parprob &\leq 
 3   \frac{\Delta_\alpha \gamma_{\uobs}}{\Delta_{r^2}} \cdot   \frac{1}{b \Phi^{-1}(1-b)} \left( \frac{1}{\sqrt{2\pi}}  \cdot \frac{1}{\alpha \Phi^{-1}(1-\alpha)}\right)^{1 / \gamma_{\uobs}^2 \cdot (1+ \eps')^2}.  
\end{align*}
\end{theorem}

\paragraph{Discussion.} Ignoring constants and lower order terms, we see that the prediction-access ratio is essentially\footnote{In particular, the $\epsilon$ and $\epsilon'$ are constants that can be chosen to be arbitrarily close to 0. Their exact value only affects the threshold $t$ that upper bounds $\alpha$ and $\Delta_{r^2}$.}
\begin{align*}
	\widetilde{\Theta}\left(\frac{\gamma_\uobs}{b} \alpha^{-1 / \gamma_{\uobs}^2}\right).
\end{align*}  
Hence, the prediction-access ratio in the probit case is significantly \emph{larger} than in the linear case. Since $1 / \gamma_\uobs^2 > 1$, the ratio behaves like $\alpha^{-c}$ for some $c > 1$ rather than $\alpha^{-1}$ as before. 
This implies that the relative value of prediction is \emph{smaller} in this discrete case and one should be relatively \emph{more} willing to expand access (even if the predictor explains only a very small fraction of the outcome variable). 

As motivated previously, the benefits of access outweigh those of prediction if the cost-benefit ratio, defined below for the probit model, is greater than 1. 
\begin{align}
\equationlabel{eq:cost_benefit_probit}
\parprob \times \frac{C_{r^2}(\Delta_{r^2})}{C_\alpha(\Delta_\alpha)}.
\end{align}
 Plugging in our bounds from \theoremref{theorem:par_probit}, and again ignoring lower order terms, the cost-benefit ratio (\equationref{eq:cost_benefit_probit}) is at least 1 if:  
\begin{align}
\equationlabel{eq:probit_cutoff}
	\gamma_{\uobs} \gtrsim \frac{C_\alpha(\Delta_\alpha)}{C_{r^2}(\Delta_{r^2})}  \alpha^{1 / \gamma^2_\uobs} b.
\end{align}
Therefore, if the level of access $\alpha$ is small (and in particular less than $\Pr[w_i > 0] = b$), as long as the costs of expanding access are not several orders of magnitude higher than those of improving prediction, expanding access is almost always the cost-efficient avenue of improvement. 

\paragraph{Visualizations.} As before, we complement these theoretical insights with numerical simulations, now presented in \figureref{fig:par_probit}. The plot displays values of the cost benefit ratio, \equationref{eq:cost_benefit_probit}, for various choices of system parameters $(\alpha, \gamma_\obs)$ and cost ratios  $C_{r^2}(\Delta_{r^2}) / C_\alpha(\Delta_\alpha)$. 

The visualizations further illustrate the technical point made previously: In this probit model the relative value of prediction is \emph{smaller} than in the linear regression case. In particular, if the level of access $\alpha$ is well below the base rate of positive improvements, $\alpha \ll b = \Prob{w_i=1}$, the cost benefit ratio is always significantly larger than 1, \emph{regardless} of the level of prediction $\gamma_\obs$. 

Due to the nonlinear nature of the problem, the threshold at which one is indifferent between expanding prediction versus expanding access is no longer linear (as seen from the relationship deduced in \equationref{eq:probit_cutoff}). Therefore, in this model, prediction and access do not go ``hand in hand'' as before. Rather, as long as $\gamma_\obs >0$, the focus should always be to expand access until $\alpha \approx b$, at which point the focus should switch to improving the predictor.

\section*{Acknowledgments}
 We would like to thank Peter Bartlett, Matthew DosSantos DiSorbo, Cynthia Dwork, Moritz Hardt, Michael P. Kim, Esther Rolf, Aaron Roth, Matt Salganik, and Keyon Vafa, as well as the anonymous reviewers at FORC and ICML, for insightful and constructive comments on our work. Conversations with Nathaniel ver Steeg were instrumental in the initial formulation of the project.
 This research was in part supported by the Harvard Center for Research on Computation and Society.

\bibliographystyle{alpha}
\bibliography{refs}

\appendix
\section{Supporting Arguments for Linear Regression Model}
\sectionlabel{sec:lr_appendix}
The following result is a general statement regarding the optimal policy for targeting problems that aim to maximize the expected welfare in the population. Nearly identical results are well-known and have appeared in numerous places (see e.g. \citet{bhattacharya2012inferring}). We include a derivation here purely for the sake of having a self-contained presentation.  
\begin{proposition}
\proplabel{prop:general_opt_policy}
For all distributions $\cD$ over $(x,w)$ supported on a discrete set of elements, the optimal policy maximizing the social planner's targeting problem, 
 \begin{align*}
	\max_{\pi \in \{\cX_\obs \rightarrow \{0,1\}\}} \quad &\E_{(x_{i,\obs}, w_i) \sim \cD}[w_i \cdot \pi(x_{i, \obs})] \\
	\text{subject to}\quad & \E_\cD[\pi(x_{i, \obs})] \leq \alpha,
\end{align*} 
is equal to, $$\pi^*_\obs(x_{i,\obs})  = \1\{E[w_i \mid x_{i,\obs}] > F^{-1}_{\obs}(1-\alpha)\} \cdot \1\{E[w_i \mid x_{i,\obs}] >0\},$$
 where $F_\obs^{-1}$ is the quantile function for the random variable $\E[w_i\mid X_\obs]$.
\begin{align*}
\end{align*}
\end{proposition}

\begin{proof}
This proof follows from viewing the optimization problem through the lens of linear programming. We start by expanding out the expectations as sums: 
\begin{align*}
	\E [w \cdot \pi(x_{\obs})] = \sum_{x_\obs, w} w \cdot \pi(x_\obs) \cdot \Prob{x_\obs, w} &= \sum_{x_\obs} \Prob{x_\obs} \pi(x_\obs) \sum_{w} w \cdot \Prob{w \mid x} = \sum_{x_\obs} \Prob{x_\obs} \pi(x_\obs) \E[w \mid x_\obs] \\ 
\E[ \pi(x_s)] &= \sum_{x} \pi(x_\obs) \Prob{x_\obs}.
\end{align*}
Therefore the optimization problem can be equivalently written as maximizing a linear cost function, while satisfying a linear constraint:
 \begin{align*}
	\max_{\pi \in \{\cX_\obs \rightarrow \{0,1\}\}} \quad & \sum_{x_\obs} \pi(x_\obs) v(x_\obs) \\
	\text{subject to}\quad & \sum \pi(x_\obs) c(x_\obs) \leq \alpha.
\end{align*} 
Here, $v(x_\obs)=  \Prob{x_\obs} \E[w \mid x_\obs]$ and $c(x_\obs) = \Pr[x_\obs]$. Without loss of generality, we can assume that $\Pr[x_\obs] > 0$ for all $x_\obs$, otherwise the terms can be removed from both sums. 

The optimal solution to this program is to set $\pi(x_\obs) = 1$, for the $x_\obs$ that maximize the ratio $v(x_\obs) / c(x_\obs)$ until the constraint is reached. Intuitively, setting $\pi(x_\obs)= 1$ ``costs'' $c(x_\obs)$ units and ``returns'' $v(x_\obs)$ units. In this setting, this ratio is equal to the conditional expectation:
\begin{align*}
	\frac{v(x_\obs)}{c(x_\obs)} = \E[w| x_\obs].
\end{align*}
By definition of the quantile function, assigning $\pi(x_\obs)= 1$ to the $x_\obs$ with the top $\alpha$ values of $\E[w\mid x_\obs]$ is equivalent to choosing 
\begin{align*}
	\pi(x_\obs) = \1\{\E[w| X_\obs] \geq F_{\obs}^{-1}(1-\alpha)\}. 
\end{align*}
Since we can always choose to treat fewer than an $\alpha$ fraction of the units, the second indicator function ensure that each term adds a nonnegative amount to the objective.
\end{proof}

\subsection{Value Function for Linear Regression Model: Proof of \propositionref{prop:value_f_linear}}
Since $w$ is Gaussian and we assume that $\alpha < .5$, $\mu > 0$, the condition that $\E[w_i\mid x_{i,\obs}] \geq F_{\obs}^{-1}(1-\alpha)$ implies that $\E[w_i \mid x_{i,\obs}] > 0$. 

Plugging in our expression for $\pi_\obs^*$ (\lemmaref{lemma:opt_policy_linear}) into the objective function, $\E[w_i \pi_\obs^*(x_{i,\obs})]$, and using the functional form for $w_i$ from the linear regression model, we have that, 
\begin{align*}
	\E[w_i \pi_\obs^*(x_{i,\obs})] & = \E[(\langle x_{i, \obs},\; \beta_\obs \rangle + \langle x_{i, \uobs},\; \beta_\uobs \rangle + \mu) \cdot \ind\{\langle x_{i,\obs},\; \beta_s \rangle \geq \Phi^{-1}(1-\alpha)\norm{\beta_\obs}\}] \\ 
	& = \alpha \mu + \E[\langle x_{i, \obs},\; \beta_\obs \rangle \cdot \ind\{\langle x_{i,\obs},\; \beta_s \rangle \geq \Phi^{-1}(1-\alpha)\norm{\beta_\obs}\}].
\end{align*}
Here, we've used the fact that the unobserved $\langle x_{i, \uobs},\; \beta_\uobs \rangle$ and observed components $\langle x_{i, \obs},\; \beta_\obs \rangle$ are independent to conclude that, 
\begin{align*}
\E[\langle x_{i, \uobs},\; \beta_\uobs \rangle \cdot \ind\{\langle x_{i,\obs},\; \beta_s \rangle \geq \Phi^{-1}(1-\alpha)\norm{\beta_\obs}\}] = \E[\langle x_{i, \uobs},\; \beta_\uobs \rangle ] \E[ \ind\{\langle x_{i,\obs},\; \beta_s \rangle \geq \Phi^{-1}(1-\alpha)\norm{\beta_\obs}\}] = 0,
\end{align*}
since $\E[\langle x_{i, \uobs},\; \beta_\uobs \rangle] = 0$. Now, using the general identity, that holds for any random variable $Z$ and event $A$,
\begin{align*}
\E[Z \cdot \ind_A] = \E[Z \mid A] \Pr[A],
\end{align*} 
we have that, for $w_{i,s} = \langle x_{i,\obs}, \; \beta_\obs \rangle$, the expectation $\E[w_{i,\obs} \cdot \ind\{ w_{i,\obs} \geq \Phi^{-1}(1-\alpha)\norm{\beta_\obs}\}]$ is equal to 
\begin{align*}
 \E[w_{i,\obs} \mid w_{i,\obs} \geq \Phi^{-1}(1-\alpha)\norm{\beta_\obs}]  \cdot \Prob{w_{i,\obs} \geq \Phi^{-1}(1-\alpha)\norm{\beta_\obs}}. 
\end{align*}
By the Mills Ratio identity (\lemmaref{lemma:mills_ratio}), since $w_{i,\obs} \sim \cN(0, \norm{\beta_\obs}^2)$, the first term is equal to,
\begin{align*}
 \E[w_{i,\obs} \mid w_{i,\obs} \geq \Phi^{-1}(1-\alpha)\norm{\beta_\obs}] &=  \norm{\beta_\obs} \frac{\phi(\frac{\Phi^{-1}(1-\alpha)\norm{\beta_\obs}}{\norm{\beta_\obs}})}{1 - \Phi(\frac{\Phi^{-1}(1-\alpha)\norm{\beta_\obs}}{\norm{\beta_\obs}})} = \frac{\norm{\beta_\obs} \phi(\Phi^{-1}(1-\alpha))}{\alpha},
\end{align*}
since $1 - \Phi(\Phi^{-1}(1-\alpha)) = \alpha$. And, by definition of the quantile function, 
\begin{align*}
\Prob{w_{i,\obs} \geq \Phi^{-1}(1-\alpha)\norm{\beta_\obs}} = \alpha.
\end{align*}
Therefore, 
\begin{align*}
\E[w_{i,\obs} \mid w_{i,\obs} \geq \Phi^{-1}(1-\alpha)\norm{\beta_\obs}] = \norm{\beta_\obs} \phi(\Phi^{-1}(1-\alpha)),
\end{align*}
and $$\E[w_i \pi_\obs^*(x_{i,\obs})] = \alpha \mu + \norm{\beta_\obs} \phi(\Phi^{-1}(1-\alpha)).$$ 
The exact expression follows by substituting $\norm{\beta_\obs} = \gamma_\obs \norm{\beta}$ (\definitionref{def:def_r2_linear}).

\subsection{Prediction Access Ratio for Linear Regression: Proof of \theoremref{theorem:par_linear}}

Recall that the goal is to prove upper and lower bounds on the ratio, 
\begin{align*}
	\frac{\Vlin(\alpha + \Delta_\alpha, \gamma_s) - \Vlin(\alpha, \gamma_s)}{\Vlin(\alpha, \gamma_s + \Delta_{r^2}) - \Vlin(\alpha, \gamma_s)}.
\end{align*}
To do so, we deal with the numerator and the denominator separately.

\paragraph{Quantifying Prediction Improvement.} We start with the denominator. This part is relatively simple because the value function is linear in $\gamma_\obs$. Recall from \propref{prop:value_f_linear} that: 
\begin{align*}
	\Vlin(\alpha, \gamma_\obs) = \alpha \mu + \gamma_\obs \norm{\beta} \phi(\Phi^{-1}(1-\alpha)). 
\end{align*}
Therefore, 
\begin{align*}
	\Vlin(\alpha, \gamma_\obs + \Delta_{r^2}) - \Vlin(\alpha, \gamma_\obs)   = \Delta_{r^2} \norm{\beta} \phi(\Phi^{-1}(1-\alpha)). 
\end{align*}
Now, if we apply \lemmaref{lemma:phi-of-Phi} and our assumption on $\alpha$, we get that, 
\begin{align}
\equationlabel{eq:beta_bounds}
\Delta_{r^2} \norm{\beta} \alpha \Phi^{-1}(1-\alpha) \leq\Vlin(\alpha, \gamma_\obs + \Delta_{r^2}) - \Vlin(\alpha, \gamma_\obs) \leq \Delta_{r^2} \norm{\beta} \alpha \Phi^{-1}(1-\alpha) (1 + f(\alpha)),
\end{align}
where $f(\alpha) = \Phi^{-1}(1-\alpha)^2 / (\Phi^{-1}(1-\alpha)^2 - 1) - 1$ is $o(1)$ and less than 1 for $\alpha < .05$.

\paragraph{Quantifying Access Improvement.} Using the closed form expression for the value function, we get that: 
\begin{align*}
\Vlin(\alpha + \Delta_{\alpha}, \gamma_\obs) - \Vlin(\alpha, \gamma_\obs)  &= \Delta_\alpha \cdot \mu + \gamma_\obs \norm{\beta} \left(\phi(\Phi^{-1}(1-\alpha - \Delta_\alpha)) - \phi(\Phi^{-1}(1-\alpha))  \right) .
\end{align*}
Applying \lemmaref{lemma:linear_acces_expansion}, we get that: 
\begin{align*}
\frac{1}{2} \Delta_\alpha \Phi^{-1}(1-\alpha) \leq 
\phi(\Phi^{-1}(1-\alpha - \Delta_\alpha)) - \phi(\Phi^{-1}(1-\alpha))
 \leq \Delta_\alpha \Phi^{-1}(1-\alpha).
\end{align*}
Therefore, 
\begin{align}
\equationlabel{eq:linear_a_bounds}
\Delta_\alpha (\mu +  \frac{1}{2} \gamma_\obs \norm{\beta} \Phi^{-1}(1-\alpha)) \leq \Vlin(\alpha + \Delta_{\alpha}, \gamma_\obs) - \Vlin(\alpha, \gamma_\obs) \leq \Delta_\alpha (\mu +  \gamma_\obs \norm{\beta} \Phi^{-1}(1-\alpha)).
\end{align}
The statement follows from combining the bounds in \equationref{eq:linear_a_bounds} and \equationref{eq:beta_bounds}.

\begin{lemma}
\lemmalabel{lemma:linear_acces_expansion}
Define $g(\alpha) = \phi(\Phi^{-1}(1-\alpha))$, then for $\alpha + \Delta < .05$ and $\Delta \leq 4\alpha$,
\begin{align*}
	 \Delta \Phi^{-1}(1-\alpha) \geq g(\alpha + \Delta) - g(\alpha)  \geq  \frac{1}{2}\Phi^{-1}(1-\alpha) \Delta.
\end{align*}
\end{lemma}

\begin{proof}
From \lemmaref{lemma:derivative_inverse}, we have that 
\begin{align*}
	 g'(\alpha) = \Phi^{-1}(1-\alpha) \text{ and } g''(\alpha) = \frac{-1}{\phi(\Phi^{-1}(1-\alpha))}.
\end{align*}
By Taylor's theorem, there exists a value $c \in [0,1]$ such that:
\begin{align*}
	g(\alpha + \Delta) - g(\alpha) &= g'(\alpha) \Delta + \frac{1}{2} \Delta^2  g''(\alpha + c \Delta) \\ 
	&= \Delta \Phi^{-1}(1-\alpha) - \frac{1}{2} \Delta^2 \frac{1}{\phi(\Phi^{-1}(1-\alpha - c \Delta))} \\ 
	& \leq \Delta \Phi^{-1}(1-\alpha).
\end{align*}
To prove the lower bound on $g(\alpha + \Delta) - g(\alpha)$ it suffices to establish that 
\begin{align}
\equationlabel{eq:2nd_order_bound}
	\frac{1}{2} \Delta^2 \frac{1}{\phi(\Phi^{-1}(1-\alpha - c \Delta))} \leq \frac{1}{2} \Delta \Phi^{-1}(1-\alpha).
\end{align}
Or equivalently, that $	\Delta  \leq  \Phi^{-1}(1-\alpha) \phi(\Phi^{-1}(1-\alpha - c \Delta))
$. To show this, we know that by \lemmaref{lemma:phi-of-Phi}, for $\alpha +  \Delta < .15$,
\begin{align*}
	\phi(\Phi^{-1}(1-\alpha - c \Delta)) \geq (\alpha + c \Delta) \Phi^{-1}(1- \alpha - c  \Delta) \geq \sqrt{2} \alpha. 
\end{align*}
The second inequality follows from the fact that $\Phi^{-1}(1-a) \geq \sqrt{2}$ for all $a < .05$ and that $c\Delta > 0$. Therefore, in order for, $$\Delta  \leq  \Phi^{-1}(1-\alpha) \phi(\Phi^{-1}(1-\alpha - c \Delta)),$$
it suffices for $\Delta \leq 4 \alpha$. This ensures that the inequality in  \equationref{eq:2nd_order_bound} is true and concludes the proof of the lower bound.
\end{proof}

\subsection{Supporting Technical Lemmas}

\begin{lemma}[Inverse Mills Ratio]
\lemmalabel{lemma:mills_ratio}
Let $z\sim \cN(\mu, \sigma^2)$, then 
\begin{align*}
	\E[z \mid z > a] = \mu + \sigma \frac{\phi(a)}{1 - \Phi(a)}.
\end{align*}
\end{lemma}
\begin{proof}
This is a well known property of Gaussians, see e.g \cite{johnson1970continuous}.
\end{proof}

\begin{lemma}
\lemmalabel{lemma:derivative_inverse}
\begin{align*}
\frac{\partial}{\partial \alpha} \Phi^{-1}(1-\alpha) &=-\frac{1}{\phi(\Phi^{-1}(1-\alpha))}  \\ 
\frac{\partial}{\partial \alpha} \phi(\Phi^{-1}(1-\alpha)) &= \Phi^{-1}(1-\alpha)
\end{align*}
\end{lemma}
\begin{proof}
These identities are also standard, we include a proof for the sake of completeness. By definition of inverse functions, 
\begin{align*}
	\frac{\partial}{\partial x} f^{-1}(f(x)) = \frac{\partial}{\partial x} x = 1,
\end{align*}
and by the chain rule, 
\begin{align*}
	\frac{\partial}{\partial x} f^{-1}(f(x)) = f'(f^{-1}(x))(f^{-1})'(x).
\end{align*}
Combining these two we get that: 
\begin{align*}
	\frac{\partial}{\partial x} f^{-1}(x) = \frac{1}{f'(f(x))} \text{ and } \frac{\partial}{\partial x} f^{-1}(1 - x) = - \frac{1}{f'(f(x))}.
\end{align*}
Since $\Phi' = \phi$ ($\Phi$ is the CDF and $\phi$ is a PDF), then we can use the second identity above to conclude that:
\begin{align*}
\frac{\partial}{\partial \alpha} \Phi^{-1}(1-\alpha) &=-\frac{1}{\phi(\Phi^{-1}(1-\alpha))}.
\end{align*}
For the second calculation, we use the fact that $\phi'(x)=-x\phi(x)$ as well as the result from the first part. By the chain rule, 
\begin{align*}
	\frac{\partial}{\partial \alpha} \phi(\Phi^{-1}(1-\alpha)) &= - \Phi^{-1}(1-\alpha) \phi(\Phi^{-1}(1-\alpha)) \left(\frac{\partial}{\partial \alpha} \Phi^{-1}(1-\alpha) \right) \\
	& = \Phi^{-1}(1-\alpha).
\end{align*}
\end{proof}

\begin{lemma}
\lemmalabel{lemma:gaussian_tails}
Let $z \sim \cN(0,1)$, then, for all $t > 0$: 
\begin{align*}
	\frac{\phi(t)}{t}(1 - t^{-2})\leq \Pr(z \geq t) \leq \frac{\phi(t)}{t}.
\end{align*}
\end{lemma}
\begin{proof}
These bounds appear in numerous sources, this particular result is drawn from a writeup from Bo Waggoner \citep{waggoner}.
\end{proof}

\begin{lemma}
\lemmalabel{lemma:phi-of-Phi}
For all $\alpha < .15$,
\begin{align*}
\alpha	 \Phi^{-1}(1-\alpha)\leq \phi(\Phi^{-1}(1-\alpha)) \leq  \Phi^{-1}(1-\alpha) \alpha (1 + f(\alpha)),
\end{align*}
where $f(\alpha) = \Phi^{-1}(1-\alpha)^2 / (\Phi^{-1}(1-\alpha)^2 - 1) - 1$ is o(1) and less than 1 for all $\alpha < .05$.
\end{lemma}
\begin{proof}
The main idea is to use \lemmaref{lemma:gaussian_tails} and set $t = \Phi^{-1}(1-\alpha)$. By definition of the quantile function, if $z$ is a standard normal:
\begin{align*}
	\Pr(z \geq \Phi^{-1}(1-\alpha)) = \alpha.
\end{align*}
From our assumption on $\alpha$, $t = \Phi^{-1}(1-\alpha) > 0$ and hence applying \lemmaref{lemma:gaussian_tails} we get that, 
\begin{align*}
	\frac{\phi(\Phi^{-1}(1-\alpha))}{\Phi^{-1}(1-\alpha)} \left(1 - \frac{1}{\Phi^{-1}(1-\alpha)^2} \right)\leq \alpha \leq \frac{\phi(\Phi^{-1}(1-\alpha))}{\Phi^{-1}(1-\alpha)}. 
\end{align*}
The final statement follows from rearranging these inequalities.
\end{proof}

\begin{proposition}
\proplabel{prop:Phi-inverse}
\begin{align*}
	\Phi^{-1}(1-\alpha)  = \Theta(\sqrt{\log(1/\alpha)}) \text{ as } \alpha \rightarrow 0
\end{align*}
\end{proposition}
\begin{proof}
From the definition of the quantile function, we have that
\begin{align*}
	\Phi^{-1}(1-\alpha) = \argmax_{t} \Pr(z \geq t) \leq \alpha.
\end{align*}
From the \lemmaref{lemma:gaussian_tails}, for $t \geq \sqrt{2}$,
\begin{align*}
	\frac{\phi(t)}{2t}  \leq \Pr[z \geq t] \leq \frac{\phi(t)}{t}.
\end{align*}  
Therefore, $\Phi^{-1}(1-\alpha) \geq t_1$ where $t_1$ solves,
\begin{align*}
	\frac{\phi(t_1)}{t_1} = \alpha,
\end{align*}
and $\Phi^{-1}(1-\alpha) \leq t_2$ where $t_2$ solves, 
\begin{align*}
	\frac{\phi(t_2)}{2t_2} = \alpha.
\end{align*}
Using the definition of $\phi$, it suffices to solve for the values of $t$ that solve the equation,
\begin{align*}
	\frac{1}{\sqrt{2\pi}} \exp(-t^2 / 2) \frac{1}{c_0 t} = \alpha,
\end{align*}
for $c_0 \in \{1,2\}$. Moving everything into logs, this becomes
\begin{align*}
	\frac{t^2}{2} - \log(\frac{1}{\sqrt{2\pi}}) - \log(\frac{1}{c t}) =  \log(1/\alpha).
\end{align*}
For small values of $\alpha$, the $t^2$ term dominates and we get that 
\begin{align*}
 t \approx \sqrt{\log(1/\alpha)}.
\end{align*}
\end{proof}

\section{Supporting Arguments for Probit Regression Model}

\subsection{Value Function for Probit Regression: Proof of \lemmaref{lemma:opt_policy_probit}}
Using the lemma regarding the functional form for the optimal targeting policy (\lemmaref{lemma:opt_policy_probit}), and the definition of $w_i$ from the probit model, we can expand out the value of the optimal policy as follows:
\begin{align}
  \E[w_i \pi^*_\obs(x_\obs)] = 
	& = \E[\1\{x_\obs^\top \beta_\obs +x_{\uobs}^\top \beta_{\uobs} + \mu > 0\} \cdot \1 \{x_\obs^\top \beta_\obs + \mu \geq \mu + \norm{\beta_\obs} \Phi^{-1}(1-\alpha) ] \notag\\
	& = \E[\1\{x_\obs^\top \beta_\obs + x_{\uobs}^\top \beta_{\uobs} + \mu > 0\} \cdot \1 \{x_\obs^\top \beta_\obs  \geq  \norm{\beta_\obs} \Phi^{-1}(1-\alpha)]. \equationlabel{eq:first_decomp_discrete}
\end{align}
Using the fact that the random variables $x_\obs^\top \beta_\obs \sim \cN(0, \norm{\beta_\obs}^2)$ and $x_{\uobs}^\top \beta_{\uobs} \sim \cN(0, \norm{\beta_\uobs}^2)$  are independent and Gaussian, we can rewrite these as, 
\begin{align*}
	x_\obs^\top \beta_\obs = \norm{\beta_\obs} z_\obs \text{ and } x_{\uobs}^\top \beta_{\uobs} = \norm{\beta_{\uobs}} z_{\uobs},
\end{align*}
where $z_\obs, z_{\uobs} \sim \cN(0,1)$ are i.i.d. Adopting this notation, we can rewrite the previous expression for the value function in \equationref{eq:first_decomp_discrete} as: 
\begin{align*}
 &= \E[\1\{\norm{\beta_\obs} z_\obs + \norm{\beta_{\uobs}} z_{\uobs}
 + \mu > 0\} \cdot \1 \{\norm{\beta_\obs} z_\obs  \geq  \norm{\beta_\obs} \Phi^{-1}(1-\alpha) ]\\ & = \E[\1\{\norm{\beta_\obs} z_\obs + \norm{\beta_{\uobs}} z_{\uobs}
 + \mu > 0\} \cdot \1 \{ z_\obs  \geq   \Phi^{-1}(1-\alpha) ]  \\
 &= \Pr[\1\{\norm{\beta_\obs} z_\obs + \norm{\beta_{\uobs}} z_{\uobs}
 + \mu > 0, \; z_\obs  \geq   \Phi^{-1}(1-\alpha)] \\ 
 & = \Pr[z_{\uobs} \geq \frac{- \norm{\beta_\obs} z_\obs - \mu}{\norm{\beta_{\uobs}}}, \; z_\obs \geq \Phi^{-1}(1-\alpha)] \\ 
 & = \int_{\Phi^{-1}(1-\alpha)}^\infty \int_{\frac{- \norm{\beta_\obs} z_\obs - \mu}{\norm{\beta_{\uobs}}}}^\infty \phi(z_{\uobs}) \phi(z_\obs) dz_{\uobs} dz_{\obs} \\ 
 & =  \int_{\Phi^{-1}(1-\alpha)}^\infty \left(1 - \Phi\left(\frac{- \norm{\beta_\obs} z_\obs - \mu}{\norm{\beta_{\uobs}}}\right)\right) \phi(z_\obs).
\end{align*}
Here, we've again used $\phi$ denote the pdf for a standard Gaussian. Using the identity, $1 - \Phi(-a) = \Phi(a)$, this last expression is equal to the following integral: 
\begin{align*}
	\int_{\Phi^{-1}(1-\alpha)}^\infty \Phi\left(\frac{\norm{\beta_\obs} z_\obs + \mu}{\norm{\beta_{\uobs}}}\right) \phi(z_\obs) dz_{\obs}.
\end{align*}
The exact statement follows from substituting $\norm{\beta_\obs} = \gamma_\obs \norm{\beta}$ and $\norm{\beta_\uobs} = \gamma_{\uobs} \norm{\beta}$.

\subsection{Prediction Access Ratio for Probit Regression: Proof of \theoremref{theorem:par_probit}}
Recall from the definition of the probit model that,
\begin{align*}
	\Pr[\langle x, \beta\rangle  + \mu > 0] = b \in (0,1),
\end{align*}
and hence, $\mu \norm{\beta}^{-1} = \Phi^{-1}(b) = - \Phi^{-1}(1-b)$.

\paragraph{Quantifying Improvements by Expanding Access.} By Taylor's theorem, there exists some $c \in (0,1)$ such that:
\begin{align*}
\Vprob(\alpha + \Delta_\alpha, \gamma_\obs) - \Vprob(\alpha , \gamma_\obs) = \Delta_\alpha \cdot \frac{\partial}{\partial \alpha} \Vprob(\alpha + c \Delta_\alpha, \gamma_\obs).
\end{align*}
The last term is the derivative of the value function evaluated at the point $(\alpha + c \Delta_\alpha, \gamma_\obs)$.
From, \lemmaref{lemma:discrete_alpha_derivative}, we know that 
\begin{align*}
\frac{\partial}{\partial \alpha} \Vprob(\alpha + c \Delta_\alpha, \gamma_\obs) = \Phi\left(\frac{\gamma_\obs \Phi^{-1}(1-\alpha - c \Delta_\alpha) + \mu \norm{\beta}^{-1}}{\gamma_{\uobs}}  \right).
\end{align*}
Using the identity, $\mu \norm{\beta}^{-1} = - \Phi^{-1}(1-b)$,  we can rewrite this as 
\begin{align}
\label{eq:intermediate_phi}
\Phi\left(\frac{\gamma_\obs \Phi^{-1}(1-\alpha - c \Delta_\alpha) + \mu \norm{\beta}^{-1}}{\gamma_{\uobs}}  \right)  = \Phi\left( \frac{1}{\gamma_{\uobs}}\big(\gamma_\obs\Phi^{-1}(1- \alpha - c\Delta_\alpha) - \Phi^{-1}(1-b)\big)\right). 
\end{align}
Assuming that $\Delta_\alpha \leq \alpha$ and that $\alpha$ is small enough so that $\gamma_\obs \Phi^{-1}(1-2\alpha) \geq \Phi^{-1}(1-b)$, then,
\begin{align*}
	\gamma_\obs\Phi^{-1}(1- \alpha - c\Delta_\alpha) \geq \gamma_\obs \Phi^{-1}(1-2\alpha) \geq \Phi^{-1}(1-b),
\end{align*}
and we have that the expression in \eqref{eq:intermediate_phi} is of the form $\Phi(x)$ for some $x \geq 0$. Therefore, by properties of the Gaussian CDF, 
\begin{align*}
\frac{\partial}{\partial \alpha} \Vprob(\alpha + c \Delta_\alpha, \gamma_\obs) \in [1/2, 1].
\end{align*}
Hence, 
\begin{align}
\label{eq:probit_acces_bounds}
	\frac{1}{2} \Delta_\alpha \leq \Vprob(\alpha + \Delta_\alpha, \gamma_\obs) - \Vprob(\alpha , \gamma_\obs) \leq \Delta_\alpha.
\end{align}

\paragraph{Quantifying Improvements via Prediction.} We use the same strategy as in the previous part. By Taylor's theorem, there exists some $c \in (0,1)$ such that:
\begin{align*}
\Vprob(\alpha, \gamma_\obs + \Delta_{r^2}) - \Vprob(\alpha , \gamma_\obs) = \Delta_{r^2} \cdot \frac{\partial}{\partial \gamma_\obs} \Vprob(\alpha, \gamma_\obs + c \Delta_{r^2}).
\end{align*}
Now, let $\gamma_{\obs,c} = \gamma_\obs + c \Delta_{r^2}$ and $\gamma_{\uobs,c} = \sqrt{1 - \gamma_{\obs,c}^2}$. Then, by \lemmaref{lemma:discrete_derivative_gamma},  
\begin{align*}
\frac{\partial}{\partial \gamma_\obs} \Vprob(\alpha, \gamma_\obs + c \Delta_{r^2}) &= 	\frac{1}{\gamma_{\uobs, c}} \phi\left(\frac{\mu}{\norm{\beta}} \right) \phi\left(\frac{\Phi^{-1}(1-\alpha) + \mu \norm{\beta}^{-1} \gamma_{\obs,c}}{\gamma_{\uobs,c}}\right) \\ 
& = \underbrace{\frac{1}{\gamma_{\uobs, c}} \phi\left(\frac{\mu}{\norm{\beta}} \right)}_{T_1} \underbrace{\phi\left(\frac{\Phi^{-1}(1-\alpha) - \Phi^{-1}(1-b) \gamma_{\obs,c}}{\gamma_{\uobs,c}}\right)}_{T_2}.
\end{align*}
We start by analyzing bounds on $T_2$. For any $\epsilon_1 > 0$, there exists a value $t_1(\epsilon_1, \gamma_t)$, depending on $\epsilon_1$ and $\gamma_t$, such that for all $\Delta_{r^2} \leq t_1(\epsilon_1, \gamma_t)$, we have that $.5 \gamma_{\uobs} \leq (1 -\eps_1) \gamma_{\uobs} \leq \gamma_{\uobs, c} \leq \gamma_{\uobs}.$
Assuming this condition holds, since $\gamma_{\obs, c} \Phi^{-1}(1-b) >0$, we get that: 
\begin{align*}
	\frac{\Phi^{-1}(1-\alpha) - \Phi^{-1}(1-b) \gamma_{\obs,c}}{\gamma_{\uobs,c}} \leq \Phi^{-1}(1-\alpha)  \frac{1}{(1- \eps_1) \gamma_{\uobs}}.
\end{align*}
Furthermore, since $\gamma_{\uobs,c} \leq \gamma_{\uobs}$ and $\gamma_{\obs,c} \leq 1$, for any fixed $\eps_2$, there exists a value $t_2(b, \gamma_t)$, depending on $\gamma_t$ and $b$, such that for all $\alpha \leq t_2(b, \gamma_t, \eps_2)$:
\begin{align*}
	\frac{\Phi^{-1}(1-\alpha) - \Phi^{-1}(1-b) \gamma_{\obs,c}}{\gamma_{\uobs,c}} \geq 	\frac{\Phi^{-1}(1-\alpha) - \Phi^{-1}(1-b) }{\gamma_{\uobs,c}} \geq 	\frac{\Phi^{-1}(1-\alpha)}{\gamma_{\uobs}} (1- \epsilon_2).
\end{align*}
Using these last two lines, and the fact that $\phi(x') \leq \phi(x)$ for $0 < x \leq x'$, we get that 
\begin{align*}
	 \phi\left(\frac{1}{(1-\epsilon_3)} \frac{\Phi^{-1}(1-\alpha)}{\gamma_\uobs}\right) \leq  \phi\left(\frac{\Phi^{-1}(1-\alpha) - \Phi^{-1}(1-b) \gamma_{\obs,c}}{\gamma_{\uobs,c}}\right) \leq \phi\left(\frac{\Phi^{-1}(1-\alpha)}{\gamma_{\uobs}} (1- \epsilon_3)\right).
\end{align*}
where $\eps_{3} = \max\{\eps_1, \eps_2\}$. Now, by \lemmaref{lemma:k-phi-of-Phi}, for $\alpha$ smaller than some value $t_3(\alpha, \eps_3) > 0$,
\begin{align*}
T_2 \leq \phi\left(\frac{\Phi^{-1}(1-\alpha)}{\gamma_{\uobs}} (1- \epsilon_3)\right) &\leq  \frac{1}{\sqrt{2\pi}}\left( 1.01 \sqrt{2\pi}  \alpha \Phi^{-1}(1-\alpha)\right)^{\frac{1}{\gamma_{\uobs}^2}(1- \epsilon_3)^2}. 
\end{align*}
And,
\begin{align*}
T_2\geq \phi\left(\frac{1}{(1-\epsilon_3)} \frac{\Phi^{-1}(1-\alpha)}{\gamma_\uobs}\right) & \geq   \frac{1}{\sqrt{2\pi}}\left(\sqrt{2\pi}  \alpha \Phi^{-1}(1-\alpha)\right)^{\frac{(1- \epsilon_3)^{-2}}{\gamma_{\uobs}^2}} 
 \\
 & = \frac{1}{\sqrt{2\pi}}\left(\sqrt{2\pi}  \alpha \Phi^{-1}(1-\alpha)\right)^{\frac{1}{\gamma_{\uobs}^2} (1 + \eps_4)^2}, 
\end{align*}
where we have rewritten the last expression in a more convenient form by letting $\eps_4 = \eps_3 / (1-\eps_3)$.

Analyzing $T_1$ is simple. Recall that $\mu \norm{\beta}^{-1} = \Phi^{-1}(1-b)$ where $b = \Prob{w_i >0}$. Hence,
\begin{align*}
	\phi\left(\frac{\mu}{\norm{\beta}}\right) = \phi(-\Phi^{-1}(1-b)) = \phi(\Phi^{-1}(1-b)).
\end{align*}
Then for all $b < .15$, \lemmaref{lemma:phi-of-Phi} ensures that:
\begin{align*}
	 b \Phi^{-1}(1-b)\leq \phi\left(\frac{\mu}{\norm{\beta}}\right) \leq 2b \Phi^{-1}(1-b).
\end{align*} 
Furthermore, by our initial calculation, we can set $\eps_1 < .5$ so that $.5 \gamma_{\uobs} \leq \gamma_{\uobs, c} \leq \gamma_{\uobs}.$ and hence: 
\begin{align*}
	\frac{1}{\gamma_\uobs} \leq \frac{1}{\gamma_{\uobs,c}} \leq \frac{2}{\gamma_{\uobs}}.
\end{align*}
This implies that $T_1$ satisfies the bounds: 
\begin{align*}
  b \Phi^{-1}(1-b) \leq 	T_1 \leq 4 \frac{1}{\gamma_{\uobs}} b \Phi^{-1}(1-b).
\end{align*}
Combining our bounds for $T_1$ and $T_2$, we get that for appropriately small $b, \alpha, \Delta_{r^2}$, 
\begin{align*}
	\Vprob(\alpha, \gamma_\obs + \Delta_{r^2}) - \Vprob(\alpha , \gamma_\obs) &\leq  \Delta_{r^2} \frac{2}{\gamma_{\uobs}} \sqrt{\frac{2}{\pi}} \cdot b \Phi^{-1}(1-b) \cdot \left(1.01\sqrt{2\pi} \alpha \Phi^{-1}(1-\alpha)\right)^{\frac{1}{\gamma_{\uobs}^2}(1- \epsilon_3)^2}\\
	\Vprob(\alpha, \gamma_\obs + \Delta_{r^2}) - \Vprob(\alpha , \gamma_\obs)  &\geq  \Delta_{r^2} \frac{1}{\gamma_{\uobs}} \frac{1}{\sqrt{2\pi}} \cdot b \Phi^{-1}(1-b)\left(\sqrt{2\pi}  \alpha \Phi^{-1}(1-\alpha)\right)^{\frac{1}{\gamma_{\uobs}^2} (1 + \eps_4)^2}.
\end{align*}
The final statement comes from combining these inequalities with those in \eqref{eq:probit_acces_bounds} and simplifying the constants. 

\subsection{Supporting Technical Lemmas}

\begin{lemma}[Leibniz's Rule]
\lemmalabel{lemma:leibniz}
\begin{align*}
    \frac{d}{dx} \int_{a(x)}^{b(x)}f(x,t)\,dt = f(x, b(x))\cdot \frac{d}{dx}b(x) - f(x, a(x))\cdot \frac{d}{dx}a(x) +
    \int_{a(x)}^{b(x)}\frac{\partial}{\partial x}f(x,t)\,dt 
\end{align*}
\end{lemma}

\begin{lemma}
\lemmalabel{lemma:discrete_alpha_derivative}
\begin{align*}
	\frac{\partial}{\partial \alpha} \Vprob(\alpha, \gamma_\obs) = \Phi\left(\frac{\gamma_\obs \Phi^{-1}(1-\alpha) + \mu \norm{\beta}^{-1}}{\gamma_{\uobs}}  \right)
\end{align*}
\end{lemma}

\begin{proof}
By Leibniz's Rule (\lemmaref{lemma:leibniz}), the derivative has only one nonzero term: 
\begin{align*}
	\frac{\partial}{\partial \alpha} \int_{\Phi^{-1}(1-\alpha)}^\infty \Phi\left(\frac{\gamma_\obs z_\obs + \mu \norm{\beta}^{-1}}{\gamma_\uobs}\right) \phi(z_\obs) dz_{\obs} &= 
- \Phi\left(\frac{\gamma_\obs \Phi^{-1}(1-\alpha) + \mu \norm{\beta}^{-1}}{\gamma_{\uobs}}  \right) \phi(\Phi^{-1}(1-\alpha)) \left ( \frac{\partial}{\partial \alpha} \Phi^{-1}(1-\alpha)\right) 
\end{align*}
To finish, we can apply \lemmaref{lemma:derivative_inverse}, to evaluate, 
\begin{align*}
	\frac{\partial}{\partial \alpha} \Phi^{-1}(1-\alpha) = -(\phi(\Phi^{-1}(1-\alpha)))^{-1},
\end{align*}
and then simplify. 
\end{proof}

\begin{lemma}
\lemmalabel{lemma:discrete_derivative_gamma}
\begin{align*}
	\frac{\partial}{\partial \gamma_{\obs}}\Vprob(\alpha, \gamma_\obs) = 	\frac{1}{\gamma_\uobs} \phi\left(\frac{\mu}{\norm{\beta}} \right) \phi\left(\frac{\Phi^{-1}(1-\alpha) + \mu \norm{\beta}^{-1} \gamma_{\obs}}{\gamma_{\uobs}}\right)
\end{align*}
\end{lemma}

\begin{proof}
Starting with the chain rule, we have that,
\begin{align*}
	\frac{\partial}{\partial \gamma_{\obs}} 	\int_{\Phi^{-1}(1-\alpha)}^\infty \Phi\left(\frac{\gamma_{\obs} z_\obs + \mu \norm{\beta}^{-1}}{\gamma_{\uobs}}\right) \phi(z_\obs) dz_{\obs},
\end{align*}
is equal to:
\begin{align}
\label{eq:chain_rule}
 	\int_{\Phi^{-1}(1-\alpha)}^\infty  \phi\left(\frac{\gamma_{\obs} z_\obs + \mu \norm{\beta^{-1}}}{\gamma_{\uobs}}\right) \phi(z_\obs) \left( z_\obs \cdot   \frac{\partial}{\partial \gamma_{\obs}}  \left(\frac{\gamma_{\obs}}{\gamma_{\uobs}} \right) + \frac{\partial}{\partial \gamma_{\obs}}  \left( \frac{\mu \norm{\beta}^{-1}}{\gamma_{\uobs}}\right)\right).
\end{align}
Now, we can calculate the remaining derivatives explicitly,
\begin{align*}
	\frac{\partial}{\partial \gamma_{\obs}} \left(\frac{\gamma_{\obs}}{\gamma_{\uobs}} \right)&= \frac{\partial}{\partial \gamma_{\obs}} \left(\frac{\gamma_{\obs}}{\sqrt{1 - \gamma_{\obs}^2}} \right) = \frac{1}{(1 - \gamma_\obs^2)^{3/2}} = \frac{1}{\gamma_t^3} \\ 
	\frac{\partial}{\partial \gamma_{\obs}} \left(\frac{\mu \norm{\beta}^{-1}}{\gamma_{\uobs}} \right)&= \frac{\partial}{\partial \gamma_{\obs}} \left(\frac{\mu \norm{\beta}^{-1}}{\sqrt{1 - \gamma_{\obs}^2}} \right) = \frac{\mu \norm{\beta}^{-1} \gamma_\obs}{(1 - \gamma_\obs^2)^{3/2}} = \frac{\mu \norm{\beta}^{-1} \gamma_\obs}{\gamma_t^3},
\end{align*}
and we can rewrite the product of pdfs as: 
\begin{align*}
\phi(\frac{\gamma_{\obs} z_\obs + \mu \norm{\beta}^{-1}}{\gamma_{\uobs}}) \phi(z_\obs) = \phi\left( \frac{\mu}{\norm{\beta}} \right) \phi\left(\frac{z_\obs + \mu \norm{\beta}^{-1} \gamma_{\obs}}{\gamma_{\uobs}}\right).
\end{align*}

Going back to the initial expression in \eqref{eq:chain_rule} and substituting in these last three identities, we have that the derivative can be written as: 
\begin{align}
\label{eq:missing_integrals}
	 \frac{1}{\gamma_t^3} \phi\left( \frac{\mu}{\norm{\beta}} \right) 	\int_{\Phi^{-1}(1-\alpha)}^\infty  
	 \phi\left(\frac{z_\obs + \mu \norm{\beta}^{-1} \gamma_{\obs}}{\gamma_{\uobs}}\right)	
		 \left( z_\obs  + \mu \norm{\beta}^{-1} \gamma_\obs\right)dz_{\obs}.
\end{align}
These integrals we can indeed evaluate: 
\begin{align*}
	\int_{\Phi^{-1}(1-\alpha)}^\infty  
	 \phi\left(\frac{z_\obs + \mu \norm{\beta}^{-1} \gamma_{\obs}}{\gamma_{\uobs}}\right)	
		  dz_{\obs}& =   \gamma_\obs \cdot \Phi\left(\frac{z_\obs + \mu \norm{\beta}^{-1} \gamma_{\obs}}{\gamma_{\uobs}}\right)\Bigg|_{\Phi^{-1}(1-\alpha)}^\infty  \\
		  & = \gamma_{\uobs}\left( 1 - \Phi\left(\frac{\Phi^{-1}(1-\alpha) + \mu \norm{\beta}^{-1} \gamma_{\obs}}{\gamma_{\uobs}}\right)  \right).
\end{align*}
And,
\begin{align*}
\int_{\Phi^{-1}(1-\alpha)}^\infty  
	 \phi\left(\frac{z_\obs + \mu \norm{\beta}^{-1} \gamma_{\obs}}{\gamma_{\uobs}}\right)	
		  z_\obs dz_{\obs} &= - \gamma_\uobs^2 \left[ 	 \phi\left(\frac{z_\obs + \mu \norm{\beta}^{-1} \gamma_{\obs}}{\gamma_{\uobs}}\right)	+ \frac{\mu \norm{\beta}^{-1}\gamma_\obs}{\gamma_{\uobs}} \Phi\left(\frac{z_\obs + \mu \norm{\beta}^{-1} \gamma_{\obs}}{\gamma_{\uobs}}\right)  \right] \Bigg|_{\Phi^{-1}(1-\alpha)}^\infty ,
		  \end{align*}
is equal to:
		  \begin{align*}
 - \gamma_\uobs^2 \left[ - \phi\left(\frac{\Phi^{-1}(1-\alpha) + \mu \norm{\beta}^{-1} \gamma_{\obs}}{\gamma_{\uobs}}\right) + \frac{\mu \norm{\beta}^{-1}\gamma_\obs}{\gamma_{\uobs}} \left( 1 - \Phi\left(\frac{\Phi^{-1}(1-\alpha) + \mu \norm{\beta}^{-1} \gamma_{\obs}}{\gamma_{\uobs}}\right)   \right)\right].
\end{align*}
Plugging in the solutions to the integrals back into \eqref{eq:missing_integrals}, we get that the final expression is: 
\begin{align*}
	\frac{1}{\gamma_\uobs} \phi\left(\frac{\mu}{\norm{\beta}} \right) \phi\left(\frac{\Phi^{-1}(1-\alpha) + \mu \norm{\beta}^{-1} \gamma_{\obs}}{\gamma_{\uobs}}\right).
\end{align*}
\end{proof}

\begin{lemma}
Assume that $\gamma_\obs^2 + \gamma_{\uobs}^2 = 1$ and that $\gamma_{\obs}, \gamma_{\uobs}$ and both in $(0,1)$. Then, 
\lemmalabel{lemma:pdf_product}
\begin{align*}
\phi(\frac{\gamma_{\obs} z_\obs + \mu \norm{\beta}^{-1}}{\gamma_{\uobs}}) \phi(z_\obs) = \phi\left( \frac{\mu}{\norm{\beta}} \right) \phi\left(\frac{z_\obs + \mu \norm{\beta}^{-1} \gamma_{\obs}}{\gamma_{\uobs}}\right).
\end{align*}
\end{lemma}
\begin{proof}
The main idea is to expand out the exponents and complete the square:
\begin{align*}
	\phi(\frac{\gamma_{\obs} z_\obs + \mu \norm{\beta}^{-1}}{\gamma_{\uobs}}) \phi(z_\obs) & = \frac{1}{2\pi} \exp\left( -\frac{(\gamma_{\obs} z_\obs + \mu \norm{\beta}^{-1})^2}{2\gamma_{\uobs}^2}  - \frac{z_\obs^2}{2}\right) \\ 
	& =  \frac{1}{2\pi} \exp\left( -\frac{((\gamma_{\obs}^2 + \gamma_{\uobs}^2) z_\obs^2 + \mu^2 \norm{\beta}^{-2} +2\mu \norm{\beta}^{-1} \gamma_{\obs} z_{\obs} )}{2\gamma_{\uobs}^2} \right) \\ 
	& = \frac{1}{2\pi} \exp\left( -\frac{\mu^2 \norm{\beta}^{-2}}{2\gamma_{\uobs}^2}\right) \exp\left(  -\frac{(z_\obs^2  +2\mu \norm{\beta}^{-1} \gamma_{\obs} z_{\obs} \pm \mu^2  \norm{\beta}^{-1}\gamma_{\obs}^2 )}{2\sigma_{\uobs}^2} \right) \\ 
	& = \frac{1}{2\pi}  \exp\left( -\frac{\mu^2 \norm{\beta}^{-2}}{2\gamma_{\uobs}^2}\right)  \exp\left(  \frac{-(z_\obs +  \mu \norm{\beta}^{-1} \gamma_{\obs})^2 + \mu^2 \norm{\beta}^{-2} \gamma_{\obs}^2 )}{2\gamma_{\uobs}^2} \right) \\ 
	& = \frac{1}{2\pi}  \exp\left( -\frac{\mu^2 \norm{\beta}^{-2}}{2\gamma_{\uobs}^2}\right) \exp \left( - \frac{1}{2}\left( \frac{z_\obs + \mu \norm{\beta}^{-1} \gamma_{\obs}}{\gamma_{\uobs}} \right)^2 \right)
 \exp\left( \frac{\mu^2 \norm{\beta}^{-2}\gamma_{\obs}^2}{2\gamma_{\uobs}^2} \right).
\end{align*}
In the third line, we used the fact that $1 = \gamma_{\uobs}^2 + \gamma_{\obs}^2$. Using this same identity, we can further simplify two of the remaining terms as:
\begin{align*}
 \exp\left( -\frac{\mu^2 \norm{\beta}^{-2}}{2\gamma_{\uobs}^2}\right) \exp\left( \frac{\mu^2 \norm{\beta}^{-2}\gamma_{\obs}^2}{2\gamma_{\uobs}^2} \right) = \exp \left(\frac{-\mu \norm{\beta}^{-2} (1 - \gamma_\obs^2)}{\gamma_{\uobs}^2} \right) = \exp\left( -\frac{\mu^2}{\norm{\beta}^2}\right).
\end{align*}
Hence, 
\begin{align*}
\phi(\frac{\gamma_{\obs} z_\obs + \mu \norm{\beta}^{-1}}{\gamma_{\uobs}}) \phi(z_\obs) = \phi\left( \frac{\mu}{\norm{\beta}} \right) \phi\left(\frac{z_\obs + \mu \norm{\beta}^{-1} \gamma_{\obs}}{\gamma_{\uobs}}\right).
\end{align*}
\end{proof}

\begin{lemma}
\lemmalabel{lemma:k-phi-of-Phi}
For any $k, \epsilon > 0$ there exists a value $t(\eps) > 0$, that depends only $\eps$, such that for all $\alpha \leq t(\eps)$, the followings inequalities hold: 
\begin{align*}
 \frac{1}{\sqrt{2\pi}} (\sqrt{2\pi} \alpha)^{k^2} (\Phi^{-1}(1-\alpha))^{k^2}
 \leq	\phi(k \cdot \Phi^{-1}(1-\alpha)) \leq \frac{1}{\sqrt{2\pi}} \left( (1+ \epsilon)  \cdot \sqrt{2\pi} \alpha \Phi^{-1}(1-\alpha) \right)^{k^2}
\end{align*}
\end{lemma}
\begin{proof}
The proof is just a direct calculation that follows from expanding out the definition of the Gaussian PDF $\phi(\cdot)$, and then applying \lemmaref{lemma:phi-of-Phi}:
\begin{align*}
	\phi(k \cdot \Phi^{-1}(1-\alpha))  &= \frac{1}{\sqrt{2\pi}} \exp \left( \frac{-k^2 (\Phi^{-1}(1-\alpha))^2}{2} \right) \\
	&= (2\pi)^{\frac{k^2 -1}{2}} \left[\frac{1}{\sqrt{2\pi}} \exp \left( \frac{- (\Phi^{-1}(1-\alpha))^2}{2} \right) \right]^{k^2} \\
	& =  (2\pi)^{\frac{k^2 -1}{2}} \left[\phi(\Phi^{-1}(1-\alpha)) \right]^{k^2}.
\end{align*}
Moving onto the second part, for $\alpha$ small enough, we have that by \lemmaref{lemma:phi-of-Phi}, 
\begin{align*}
	\alpha \Phi^{-1}(1-\alpha) \leq \phi(\Phi^{-1}(1-\alpha)) \leq \alpha \Phi^{-1}(1-\alpha) (1 + \eps). 
\end{align*}
Therefore, 
\begin{align*}
(2\pi)^{\frac{k^2 -1}{2}} \left[\alpha \cdot  \Phi^{-1}(1-\alpha) \right]^{k^2}	\leq \phi(k \cdot \Phi^{-1}(1-\alpha)) \leq (2\pi)^{\frac{k^2 -1}{2}} \left[ (1+ \eps) \alpha \cdot  \Phi^{-1}(1-\alpha) \right]^{k^2}.
\end{align*}
The precise inequality then follows from simplifying this expression above.
\end{proof}
\section{Simulation Details}

For the visualizations in \figureref{fig:par_linear} and \figureref{fig:par_probit}, we compute the prediction access ratios numerically. Using our closed form expressions regarding the value functions $\Vlin$ and $\Vprob$ from \propositionref{prop:value_f_linear} and \propositionref{prop:value_f_probit} we compute the cost benefit ratios with $\Delta_\alpha = \Delta_{r^2} = .01$ and $\mu = 1$, $\beta=10$ for the linear case. For the probit case, we set $b=.1$ and $\Delta_\alpha = \Delta_{r^2} = 1e-3$.

\end{document}